\documentclass[12pt]{article}

\usepackage{amsmath,amsfonts,amssymb,amsthm,graphicx,cite}

\numberwithin{equation}{section}

\begin{document}

\newtheorem{theorem}{Theorem}[section]
\newtheorem{corollary}[theorem]{Corollary}
\newtheorem{lemma}[theorem]{Lemma}
\newtheorem{proposition}[theorem]{Proposition}

\newcommand{\adiffop}{A$\Delta$O}
\newcommand{\adiffops}{A$\Delta$Os}

\newcommand{\be}{\begin{equation}}
\newcommand{\ee}{\end{equation}}
\newcommand{\bea}{\begin{eqnarray}}
\newcommand{\eea}{\end{eqnarray}}
\newcommand{\sh}{{\rm sh}}
\newcommand{\ch}{{\rm ch}}
\newcommand{\einde}{$\ \ \ \Box$ \vspace{5mm}}
\newcommand{\De}{\Delta}
\newcommand{\de}{\delta}
\newcommand{\Z}{{\mathbb Z}}
\newcommand{\N}{{\mathbb N}}
\newcommand{\C}{{\mathbb C}}
\newcommand{\Cs}{{\mathbb C}^{*}}
\newcommand{\R}{{\mathbb R}}
\newcommand{\Q}{{\mathbb Q}}
\newcommand{\T}{{\mathbb T}}
\newcommand{\D}{{\mathbb D}}
\newcommand{\re}{{\rm Re}\, }
\newcommand{\im}{{\rm Im}\, }
\newcommand{\cW}{{\cal W}}
\newcommand{\cJ}{{\cal J}}
\newcommand{\cE}{{\cal E}}
\newcommand{\cA}{{\cal A}}
\newcommand{\cR}{{\cal R}}
\newcommand{\cP}{{\cal P}}
\newcommand{\cM}{{\cal M}}
\newcommand{\cN}{{\cal N}}
\newcommand{\cI}{{\cal I}}
\newcommand{\cMs}{{\cal M}^{*}}
\newcommand{\cB}{{\cal B}}
\newcommand{\cD}{{\cal D}}
\newcommand{\cC}{{\cal C}}
\newcommand{\cL}{{\cal L}}
\newcommand{\cF}{{\cal F}}
\newcommand{\cH}{{\cal H}}
\newcommand{\cS}{{\cal S}}
\newcommand{\cT}{{\cal T}}
\newcommand{\cU}{{\cal U}}
\newcommand{\cQ}{{\cal Q}}
\newcommand{\cV}{{\cal V}}
\newcommand{\cK}{{\cal K}}
\newcommand{\intR}{\int_{-\infty}^{\infty}}
\newcommand{\intI}{\int_{0}^{\pi/2r}}
\newcommand{\limp}{\lim_{\re x \to \infty}}
\newcommand{\limn}{\lim_{\re x \to -\infty}}
\newcommand{\limpn}{\lim_{|\re x| \to \infty}}
\newcommand{\diag}{{\rm diag}}
\newcommand{\Ln}{{\rm Ln}}
\newcommand{\Arg}{{\rm Arg}}

\title{A recursive construction of joint eigenfunctions for the hyperbolic nonrelativistic Calogero-Moser Hamiltonians}

\author{Martin Halln\"as \\Department of Mathematical Sciences, \\ Loughborough University, Leicestershire LE11 3TU, UK \\ and \\Simon Ruijsenaars \\ School of Mathematics, \\ University of Leeds, Leeds LS2 9JT, UK}

\date{ }

\maketitle

\begin{abstract}
We obtain symmetric joint eigenfunctions for the commuting PDOs associated to the hyperbolic Calogero-Moser $N$-particle system. The eigenfunctions are constructed via a recursion scheme, which leads to representations by multidimensional integrals whose integrands are elementary functions. We also tie in these eigenfunctions with the Heckman-Opdam hypergeometric function for the root system~$A_{N-1}$.
\end{abstract}

\tableofcontents

\section{Introduction}

In this paper we present and develop a recursive scheme to construct joint eigenfunctions for the commuting partial differential operators (PDOs) associated with the integrable $N$-particle system of hyperbolic Calogero-Moser type. This scheme gives rise to explicit integral representations of the joint eigenfunctions involving elementary integrands, cf.~Eqs.~\eqref{intFN1}--\eqref{intFN2}. The present paper may be viewed as  a nonrelativistic counterpart of our recent paper \cite{HR13} on joint eigenfunctions for the hyperbolic relativistic Calogero-Moser system. As will be seen, the multidimensional explicit integrals obtained in Appendix~C of the latter yield once again the tools to get analytic control on the formalities of the recursive construction, and more generally the present paper is organized in a way similar to~\cite{HR13}.

As is well known, the nonrelativistic hyperbolic Calogero-Moser system can be defined by the Hamiltonian
\be\label{H2}
H_2 = -\hbar^2\sum_{1\leq j_1< j_2\leq N}\partial_{x_{j_1}}\partial_{x_{j_2}} - g(g-\hbar)\sum_{1\leq j<l\leq N}\mu^2/4\sinh^2(\mu(x_j-x_l)/2),
\ee
where $\hbar$ is Planck's constant, $g>0$ a coupling constant with dimension [action], and $\mu>0$ a parameter with dimension [position]$^{-1}$. There exist $N-1$ additional independent PDOs $H_k$, $k=1,3,\ldots,N$, that together with $H_2$ form a family of pairwise commuting PDOs. They can be chosen to be of the form
\be\label{Hk}
H_1=-i\hbar \sum_{j=1}^N \partial_{x_j}, \ \ \ \ H_k = (-i\hbar)^k\sum_{1\leq j_1<\cdots<j_k\leq N}\partial_{x_{j_1}}\cdots\partial_{x_{j_k}} + {\rm l.o.},\ \ \ k>2,
\ee
where ${\rm l.o.}$ stands for terms of lower order in the partial derivatives $\partial_{x_j}$.

For arbitrary values of the coupling constant $g$, joint eigenfunctions of these PDOs were first constructed and studied by Heckman and Opdam \cite{HO87}. More precisely, they considered the problem in the context of an arbitrary root system, whereas we restrict attention to $A_{N-1}$.  In Heckman and Opdam's work the eigenfunctions were obtained using the classical method of series substitution. For given joint eigenvalues their construction for~$A_{N-1}$ gives rise to~$N!$ joint eigenfunctions, whereas the present method only yields the (essentially unique) symmetric linear combination.

In our treatment of the $N=2$ case in Section \ref{Sec3} we shall see that the symmetric eigenfunction amounts to the so-called conical (or Mehler) function specialisation of the hypergeometric function ${}_2F_1$ (cf.~Chapter 14 in \cite{Dig10}). It was known already to Mehler \cite{Meh68} that the conical function admits a representation as a product of an elementary function and an integral that amounts to a Fourier transform of a function involving two hyperbolic cosines. For our purposes, this representation is of fundamental importance. As we shall show, it can be viewed as the first step  in a recursive construction of symmetric arbitrary-$N$ joint eigenfunctions of the commuting PDOs $H_k$. The point is that the plane wave in the Fourier transform can be viewed as the $N=1$ eigenfunction, whereas the function being transformed serves as a kernel function, connecting the free $N=1$ ODO $-i\hbar d/dx$ to the interacting $N=2$ PDOs~$H_1$ and~$H_2$.

In our recent joint paper \cite{HR12}, we presented a comprehensive study of kernel functions for all of the Calogero-Moser and Toda systems of $A_{N-1}$ type. We obtained, in particular, kernel functions connecting the hyperbolic PDOs for the $N$-particle case to those for the $(N-1)$-particle case. As we shall sketch in Section 2, this enables us to set up a recursion scheme for the construction of arbitrary-$N$ joint eigenfunctions. In the $N=2$ case the pertinent kernel function amounts to a function that occurs in a special integral representation of the conical function, as highlighted in the previous paragraph. For $N>2$ a crucial ingredient is given by the kernel identities obtained in Prop.~4.5 of~\cite{HR12}. (For a pair of defining Hamiltonians, our kernel identity coincides with the hyperbolic limit of an elliptic kernel identity due to Langmann \cite{Lan06}.)

The idea that such a recursive construction might be feasible appears to date back to work by Gutzwiller \cite{Gut81}, who used it to connect eigenfunctions for the periodic and nonperiodic nonrelativistic Toda systems. Since then, the formalism has been used for a number of different systems, in particular by Gerasimov, Kharchev and Lebedev \cite{GKL04} for the $g=\hbar/2$ specialization of the hyperbolic Calogero-Moser system and for the Toda systems, and by Felder and Veselov  \cite{FV09} to construct Baker-Akhiezer type eigenfunctions for the hyperbolic Calogero-Moser system with $g$ a negative integer multiple of $\hbar$. We also note that, for the trigonometric Calogero-Moser system, a recursive eigenfunction construction of a somewhat different type was presented by Awata, Matsuo, Odake and Shiraishi \cite{AMOS95}.

A crucial preliminary for our recursive scheme is an explicit construction of the commuting Hamiltonians~\eqref{H2}--\eqref{Hk}, which has as its starting point the classical Lax matrix
\be\label{L}
L(x,p)_{jk}\equiv \de_{jk}p_j + (1-\de_{jk})\frac{i\mu g}{2\sinh\big(\mu(x_j-x_k)/2\big)},\quad j,k=1,\ldots,N.
\ee
(We note that substituting $ix_j$ for $x_j$ in $L(x,p)$ yields a Lax matrix that is slightly different from the one introduced by Moser \cite{Mos75}: To arrive at Moser's Lax matrix one should replace the function $1/\sin$ by $\cot$.) Performing the canonical quantization substitutions
\be\label{psubs}
p_j\to-i\hbar\partial_{x_j},\quad j=1,\ldots,N,
\ee
in $L(x,p)$ produces an operator-valued matrix, whose symmetric functions $\hat{\Sigma}_k(L)(x)$ are well defined, since no ordering ambiguities occur. Indeed, a term in the expansion of a principal minor of \eqref{L} that depends on $p_j$ does not depend on $x_j$. However, $\hat{\Sigma}_2(L)(x)$ differs from the Hamiltonian \eqref{H2} in that it is missing the term proportional to $\hbar$ in the potential energy.

The problem can be resolved by adding to $L$ the diagonal $N\times N$ matrix
\be\label{E}
E(x)\equiv\diag\big(z_1(x),\ldots,z_N(x)\big),
\ee
with
\be\label{zj}
z_j(x)\equiv -\frac{i\mu g}{2}\sum_{k\neq j}\coth\big(\mu(x_j-x_k)/2\big),\quad j=1,\ldots,N.
\ee
In contrast to the symmetric functions of $L$, the functions $\Sigma_k(L(x,p)+E(x))$ with $k>1$ contain products of terms that depend on both $p_j$ and $x_j$, making their canonical quantization ambiguous. In~Section 4.3 of~\cite{Rui94} it was shown that the ordering ensuring commutativity is normal ordering: the procedure of putting $x$-dependent coefficients to the left of monomials in the momentum operators $-i\hbar\partial_{x_j}$, $j=1,\ldots,N$. We let $:\hat{\Sigma}_k(L+E)(x):$ denote the normal ordered PDOs obtained from $\Sigma_k(L(x,p)+E(x))$ by performing the canonical quantization substitutions \eqref{psubs}. Introducing the weight function
\be\label{cW}
\cW(g;x)\equiv \left(\prod_{1\leq j<k\leq N}4\sinh^2\big(\mu(x_j-x_k)/2\big)\right)^{g/\hbar},
\ee
the commuting Hamiltonians $H_k$ are then given by the formula
\be\label{sim}
H_k(x) = \cW(x)^{1/2}:\hat{\Sigma}_k(L+E)(x):\cW(x)^{-1/2},\quad k=1,\ldots,N.
\ee
(Recall we have required $g>0$, so we may and will take the positive square root.) Note that~$\cW$ and~$E$ are related by
\be\label{cWE}
 \cW(x)^{1/2}\big( -i\hbar\partial_{x_j} +E_{jj}(x)\big)  \cW(x)^{-1/2} =-i\hbar\partial_{x_j},\ \ \ j=1,\ldots ,N.
 \ee

From~\eqref{H2} it is plain that the Hamiltonian $H_2$ is formally self-adjoint on $L^2(\R^N,dx)$. Although it is not obvious from \eqref{sim}, this property holds true for the PDOs $H_k$ with $k>2$, too. This is clear by inspection from the following explicit formulas for~$H_1,\ldots,H_N$:
\be\label{HJ}
H_k=J_k,\ \ \ k=1,\ldots,N,
\ee
where~$J_k$ is the PDO
\be\label{Jk}
J_k\equiv \frac{1}{(N-k)!}\sum_{0\le l\le
[k/2]}\frac{1}{2^ll!(k-2l)!}\sum_{\sigma \in S_N}\sigma
\big(h(x_1-x_2)\cdots h(x_{2l-1}-x_{2l})\hat{p}_{2l+1}\cdots \hat{p}_k\big),
\ee  
with
\be\label{ph}
\hat{p}_j\equiv -i\hbar\partial_{x_j},\ \ \ j=1,\ldots, N,\ \ \ h(x)\equiv g(\hbar-g)\mu^2/4\sinh^2(\mu x/2). 
\ee

The PDOs $J_k$ given by~\eqref{Jk} (and their elliptic generalizations) were first studied by Olshanetsky and Perelomov~\cite{OP83}, as quantizations of classical Hamiltonians introduced in~\cite{SK75, Woj77}. Their equality to the Hamitonians~$H_k$ given by~\eqref{sim}, as encoded in~\eqref{HJ}, again follows from~Section 4.3 of~\cite{Rui94}. More specifically, this reference dealt with the elliptic versions of the commuting PDOs. Using a uniqueness result by Oshima and Sekiguchi \cite{OS95}, it was shown in~{\it loc.~cit.} that the elliptic PDO $H_k$ can be expressed as a linear combination of the elliptic  commuting PDOs~$J_1,J_2,\ldots,J_k$ and a constant. (Cf.~Eq.~(4.86) in~{\it loc.~cit.}) For the hyperbolic case, however, one can take all particle distances to infinity to deduce that the linear combination reduces to~\eqref{HJ}. Indeed, $L$~\eqref{L} becomes diagonal for~$|x_j-x_k|\to\infty$, so this readily follows from~\eqref{sim}--\eqref{cWE} and~\eqref{Jk}--\eqref{ph}.   

We are now in a position to give a more precise sketch of our main results. To achieve analytic control over the recursive scheme, it will be important to allow complex~$g$. However, we need to impose the restriction 
\be\label{gres}
\re g\geq\hbar.
\ee
 This constraint is presumably stronger than necessary, with $\re g>0$ being our conjectured necessary and sufficient condition for getting absolute convergence of the pertinent integrals (and hence analyticity in~$g$ for~$\re g$ positive). This can be verified to be the case for small $N$, but our recursive arguments hinge on the explicit evaluation of integrals obtained in Appendix~C of our previous paper~\cite{HR13}, and these integrals cannot be used for arbitrary $N$ when~$\re g$ varies over~$(0,\hbar)$.

 With~\eqref{gres} in force, we can invoke the latter integrals to proceed recursively in the number of variables $(x_1,\ldots,x_N)$. In this way we construct joint eigenfunctions $\Psi_N((x_1,\ldots,x_N),(p_1,\ldots,p_N))$ of the Hamiltonians $H_k$ with eigenvalues given by
\be\label{PsiEqs}
H_k(x)\Psi_N(x,p) = S_k(p)\Psi_N(x,p),\quad k=1,\ldots,N,
\ee
where
\be
S_k(p)\equiv \sum_{1\leq j_1<\cdots<j_k\leq N}p_{j_1}\cdots p_{j_k}
\ee
is the $k$th elementary symmetric function of the momenta $p_1,\ldots,p_N$. Moreover, we obtain certain analyticity features referring to the `geometric' and `spectral' variables~$x$ and~$p$, and the coupling parameter~$g$.

To ease the notation it is actually more convenient to work with the dimensionless quantities
\be\label{dLess}
\lambda\equiv g/\hbar,\quad (t_1,\ldots,t_N)\equiv (\mu x_1/2,\ldots\mu x_N/2),\quad (u_1,\ldots,u_N)\equiv (2p_1/\hbar\mu,\ldots,2p_N/\hbar\mu),
\ee
and rewrite $\Psi_N(x,p)$ as
\begin{multline}\label{PsiN}
\Psi_N\big(g;(x_1,\ldots,x_N),(y_1,\ldots,y_N)\big)\\ = W(g/\hbar;\mu x/2)^{1/2}F_N\big(g/\hbar;(\mu x_1/2,\ldots,\mu x_N/2),(2p_1/\hbar\mu,\ldots,2p_N/\hbar\mu)\big),
\end{multline}
where (cf.~\eqref{cW})
\be\label{cWW}
W(\lambda;t)\equiv \cW(\lambda \hbar;2t/\mu)=\prod_{1\leq j<k\leq N}[4\sinh^2(t_j-t_k)]^{\lambda}.
\ee
As a result of our construction, we arrive at simple and explicit integral representations of the similarity-transformed joint eigenfunctions  $F_N$. In particular, introducing $N(N+1)/2$ variables $t_{nj}$ with $n=1,\ldots,N$ and $j=1,\ldots,n$, and  identifying $t_{Nj}$ with~$t_j$ for $j=1,\ldots,N$, we obtain
\begin{multline}\label{intFN1}
F_N(\lambda;t,u) =  \int_{\R^{N(N-1)/2}}\prod_{n=1}^{N-1}\frac{\prod_{1\leq j<k\leq n}[4\sinh^2(t_{nj}-t_{nk})]^\lambda}{n!\prod_{j=1}^{n+1}\prod_{k=1}^n\big[2\cosh(t_{n+1,j}-t_{nk})\big]^\lambda}\\ \times \exp\left(i\sum_{n=1}^Nu_n\left(\sum_{j=1}^nt_{nj}-\sum_{j=1}^{n-1}t_{n-1,j}\right)\right)\prod_{n=1}^{N-1}\prod_{j=1}^n dt_{nj}.
\end{multline}

An alternative way of writing this integral that makes its recursive structure more apparent involves integrations over `Weyl chambers'~$G_1,\ldots, G_{N-1}$ given by
\be\label{Gk}
G_k\equiv \{ x\in\R^k\mid x_k<x_{k-1}<\cdots <x_1\}.
\ee
It reads
\bea\label{intFN2}
F_N(\lambda;t,u) & = & \exp\Big(iu_N\sum_{j=1}^N t_j\Big)\prod_{n=1}^{N-1} \int_{G_n} dt_{n1}\cdots dt_{nn}  \exp\big(i(u_n-u_{n+1})( t_{n1}+\cdots +t_{nn})\big)
 \nonumber \\
 &  & \times\frac{ 
\prod_{1\leq j<k\leq n}[2\sinh(t_{nj}-t_{nk})]^{2\lambda}}{\prod_{j=1}^{n+1}\prod_{k=1}^n\big[2\cosh(t_{n+1,j}-t_{nk})\big]^\lambda} .
\eea
As will be seen in the main text, in these representations we can allow $t$ and $u$ to vary over certain subsets of $\C^{N}$, but  to ease the exposition in this introduction we choose $u$ real and let $t_j$ vary over the strip~$\im t_j\in (-\pi/2,\pi/2)$. Also,  $\lambda$ is restricted via~\eqref{gres} and~\eqref{dLess}. Thus we wind up with the set  
\be\label{varres}
\{ (\lambda,t,u)\in\C\times\C^{N}\times \R^{N}\mid \re \lambda \ge 1,\ \ |\im t_j|<\pi/2,\ \ j=1,\ldots,N\},
\ee
on which the integrals will be proved to converge absolutely. 
(Note that in~\eqref{intFN2} all of the implied logarithms in the numerator have positive arguments thanks to the ordering in~\eqref{Gk}, so they may and will be chosen real.) This yields analyticity in $t$ in the relevant polystrip and in $\lambda$ for~$\re \lambda > 1$.

We would like to mention that Matsuo \cite{Mat92} sketched a procedure for obtaining a representation of joint eigenfunctions for the hyperbolic nonrelativistic Calogero-Moser Hamiltonians from so-called hypergeometric solutions of trigonometric Knizhnik-Zamolodchikov equations of $A_{N-1}$ type. It is of interest that the resulting representation bears some resemblance to the representations \eqref{intFN1}--\eqref{intFN2}. More specifically, it is also given as an integral over $N(N-1)/2$ variables with part of the integrand having the same type of product structure as above. On the other hand, the integration contour and part of the integrand are left unspecified in~{\it loc.~cit.} More recently, Borodin and Gorin \cite{BG13} presented a representation of eigenfunctions for the Hamiltonian $H_1^2-2H_2$ in terms of an integral over a certain $N(N-1)/2$-dimensional polytope whose integrand again has the same type of product structure as in \eqref{intFN1}--\eqref{intFN2}. 

 We conclude this introduction by sketching the content and organisation of the paper in more detail. In Section 2 we recall the relevant kernel functions and corresponding   identities, and present the recursive construction in a formal fashion (i.e.~without worrying about convergence of integrals, etc.). As we shall see, the key to obtain the explicit eigenvalues in~\eqref{PsiEqs} is the following recurrence for the elementary symmetric functions $S_k^{(M)}$ of~$M$ nonzero numbers~$a_1,\ldots,a_M$:
\be\label{Srec}
\sum_{l=0}^ka_M^l\binom{M-k+l}{l}S_{k-l}^{(M-1)}(a_1-a_M,\ldots,a_{M-1}-a_M) = S_k^{(M)}(a_1,\ldots,a_M),
\ee
where $M\geq 1$, $k=1,\ldots,M$, and
\be
S_M^{(M-1)}\equiv 0,\quad S_0^{(M-1)}\equiv 1.
\ee
It will transpire from Section~2 that the applicability of this (slightly nonobvious, but easily checked) recurrence to the problem at hand hinges on using the above matrix $L(x,p)+E(x)$ to express the commuting PDOs. In that connection we should add that a direct use of the explicit formulas~\eqref{HJ}--\eqref{ph} to establish the desired recursive structure appears intractable.  

In Section~3 we consider the case~$\lambda=1$, for which we get (cf.~\eqref{HJ}--\eqref{ph} and~\eqref{dLess})
\be
H_k=(-i\hbar \mu/2)^k\sum_{1\le j_1<\cdots <j_k\le N}\  \frac{\partial}{\partial t_{j_1}}\cdots \frac{\partial}{\partial t_{j_k}},\ \ \ k=1,\ldots,N.
\ee
This $\lambda$-choice corresponds to free cases of the hyperbolic relativistic Calogero-Moser system studied in Section~3 of our recent joint paper~\cite{HR13}, for which the commuting Hamiltonians reduce to analytic difference operators with constant coefficients. Furthermore, the associated kernel functions are basically equal to those for $\lambda=1$, so that we can take over the findings from~{\it loc.~cit.} by a suitable reparametrization. In particular, the joint eigenfunction is proportional to the kernel of the multivariate sine transform.

In Section 4 we focus on the analytic aspects of the first step of the scheme, for which we need only require~$\re \lambda >0$.  This step leads from the free one-particle plane-wave eigenfunction to the interacting two-particle eigenfunction, which amounts to a conical function (after removal of a center-of-mass factor). We obtain some of its properties, using arguments that can be generalised to the arbitrary-$N$ case. The results are encoded in Props.~\ref{h2Prop}--4.4. They consist of a holomorphy domain, the joint eigenvalue equations and uniform decay bounds for real and (suitably restricted) complex~$u$, resp.

In Section 5 we consider the step from $N=2$ to $N=3$ for $\re\lambda>0$. Although this leads to novel difficulties, we are still able to prove the counterparts of Props.~\ref{h2Prop}--4.4. The main tool for obtaining the bounds encoded in Props.~\ref{F3bProp}--5.4 is the explicit evaluation of certain integrals we encountered in Appendix~C  of~\cite{HR13}.  

Once our arguments leading to Props.~\ref{h3Prop}--5.4 are well understood, it is no longer hard to understand the inductive step treated in Section 6, yielding our main results Theorems~6.1--6.4, for which we require $\re\lambda>1$. Again, this relative simplicity hinges on the explicit integrals obtained in~\cite{HR13}.

In Section 7 we establish the connection to previous work by Heckman and Opdam~\cite{HO87,Opd93}. More specifically, we show that the `center-of-mass' eigenfunction $F_N(\lambda;t,u)$, $\sum t_j=\sum u_j=0$, is proportional to their hypergeometric function associated with the root system $A_{N-1}$.

\section{Formal structure of the recursion scheme}\label{Sec2}
We begin this section by transforming the eigenvalue equations \eqref{PsiEqs} into an equivalent set of equations for the functions $F_N$. It is convenient to slightly modify the  PDOs $:\hat{\Sigma}_k(L+E)(x):$ by switching to dimensionless counterparts. First, we renormalise the Lax matrix $L$ and the diagonal matrix $E$ by introducing
\be\label{cL}
\cL(t,u)_{jk}\equiv \frac{2}{\hbar\mu}L(2t/\mu,\hbar\mu u/2)_{jk} = \de_{jk}u_j + (1-\de_{jk})\frac{i\lambda}{\sinh(t_j-t_k)},
\ee
and
\be\label{cE}
\cE(t)\equiv \diag\big(w_1(t),\ldots,w_N(t)\big)
\ee
with
\be\label{wj}
w_j(t)\equiv \frac{2}{\hbar\mu}z_j(2t/\mu) = -i\lambda\sum_{k\neq j}\coth(t_j-t_k),
\ee
cf.~\eqref{L} and \eqref{E}--\eqref{zj}. Then we let $:\hat{\Sigma}_k(\cL+\cE)(t):$ denote the normal ordered PDOs obtained from $\Sigma_k(\cL(t,u)+\cE(t))$ by performing the substitutions
\be\label{usubs}
u_j\to -i\partial_{t_j},\quad j=1,\ldots,N.
\ee
From \eqref{sim} and \eqref{PsiN}--\eqref{dLess} we then see that \eqref{PsiEqs} is equivalent to
\be
:\hat{\Sigma}_k(\cL+\cE)(t):F_N(t,u) = S_k(u)F_N(t,u),\quad k=1,\ldots,N.
\ee

We also have occasion to use the dimensionless Hamiltonians
\be\label{sim2}
\cH_k(t) \equiv W(t)^{1/2}:\hat{\Sigma}_k(\cL+\cE)(t):W(t)^{-1/2},\quad k=1,\ldots,N.
\ee
Using~\eqref{sim}--\eqref{ph}, we see that they are explicitly given by 
\be\label{cHk}
\cH_k =
\frac{1}{(N-k)!}\sum_{0\le l\le
[k/2]}\frac{1}{2^ll!(k-2l)!}\sum_{\sigma \in S_N}\sigma
\big(h_r(t_1-t_2)\cdots h_r(t_{2l-1}-t_{2l})\hat{u}_{2l+1}\cdots \hat{u}_k\big),
\ee  
with
\be\label{uhr}
\hat{u}_j\equiv -i\partial_{t_j},\ \ \ j=1,\ldots, N,\ \ \ h_r(t)\equiv \lambda(1-\lambda )/\sinh^2t. 
\ee

With this modification of the PDOs in force, we continue to recall the relevant kernel functions and identities from Subsection IV\,A.2 of \cite{HR12}. The first kernel function is
\be\label{cK}
\cK(\lambda;t,s)\equiv \prod_{j,k=1}^N\big[2\cosh(t_j-s_k)\big]^{-\lambda}.
\ee
It satisfies the identities
\be
\big(:\hat{\Sigma}_k(\cL+\cE)(t): - :\hat{\Sigma}_k(\cL+\cE)(-s):\big)\cK(t,s) = 0,\quad k=1,\ldots,N,
\ee
so it connects the $N$-variable PDOs to themselves. We expect that this kernel function will be important to obtain further properties of the joint eigenfunctions $F_N$.  However, this is beyond the scope of the present paper. Here, we need the kernel identities relating the PDOs in $N$ variables to the PDOs in $N-1$ variables. In \cite{HR12} such identities were obtained in two ways: first, by multiplying $\cK(t,s)$ by a suitable exponential factor and then sending $s_N$ to infinity, and second, by exploiting a limit from the relativistic hyperbolic Calogero-Moser case. In order to avoid ambiguities we shall henceforth indicate the number of variables $N$ on which a PDO or kernel function depends, and also use the superscript $\sharp$ to indicate that the first argument of a kernel function depends on one more variable than the second one.
With these conventions in effect, we recall from Prop.~4.5 in \cite{HR12} that the kernel function
\be\label{cKs}
\cK_N^\sharp(\lambda;t,s)\equiv \prod_{j=1}^N\prod_{k=1}^{N-1}\big[2\cosh(t_j-s_k)\big]^{-\lambda},\quad N>1,
\ee
satisfies the eigenfunction identity
\be\label{key1}
:\hat{\Sigma}_N^{(N)}(\cL+\cE)(t):\cK_N^\sharp(t,s) = 0,
\ee
and the kernel identities
\be\label{key2}
\left(:\hat{\Sigma}_k^{(N)}(\cL+\cE)(t): - :\hat{\Sigma}_k^{(N-1)}(\cL+\cE)(-s):\right)\cK_N^\sharp(t,s) = 0,\quad k=1,\ldots,N-1.
\ee

Having assembled the necessary ingredients, we are ready to describe the `calculational' crux of the recursive scheme. Assume we have obtained a joint eigenfunction $F_{N-1}((t_1,\ldots,t_{N-1}),(u_1,\ldots,u_{N-1}))$ of the PDOs $:\hat{\Sigma}_k^{(N-1)}(\cL+\cE)(t):$, with eigenvalues given by
\be\label{recass}
:\hat{\Sigma}_k^{(N-1)}(\cL+\cE)(t):F_{N-1}(t,u) = S_k^{(N-1)}(u)F_{N-1}(t,u),\quad k=1,\ldots,N-1,
\ee
where $S_k^{(M)}(a_1,\ldots,a_M)$ denotes the elementary symmetric functions of the $M$ numbers $a_1,\ldots,a_M$. Consider the function $F_N(t,u)$ with arguments $t,u\in\C^N$, given formally by
\be\label{FN}
F_N(t,u)\equiv \frac{e^{iu_N\sum_{j=1}^Nt_j}}{(N-1)!}\int_{\R^{N-1}}dsW_{N-1}(s)\cK_N^\sharp(t,s)F_{N-1}(s,(u_1-u_N,\ldots,u_{N-1}-u_N)).
\ee
For now, we do not address the convergence of the integral, but we do restrict attention to positive $\lambda $ to prevent manifest divergencies. We also assume that we are allowed to differentiate any number of times under the integral sign.

Acting with $:\hat{\Sigma}_k^{(N)}(\cL+\cE)(t):$ on $F_N(t,u)$ and shifting through the plane wave up front, the PDO is transformed into $:\hat{\Sigma}_k^{(N)}(\cL+\cE+u_N\mathbf{1}_N)(t):$, where $\mathbf{1}_N$ denotes the $N\times N$ identity matrix. Making use of the expansion
\be\label{Sexp}
:\hat{\Sigma}_k^{(N)}\big(\cL+\cE+u_N\mathbf{1}_N\big)(t):\ =\ \sum_{l=0}^ku_N^l\binom{N-k+l}{l}:\hat{\Sigma}_{k-l}^{(N)}(\cL+\cE)(t):,
\ee
we can act on the kernel function and apply \eqref{key1}--\eqref{key2}. Using formal self-adjointness on $L^2(\R^{N-1},W_{N-1}(s)ds)$ of the PDOs at hand (which follows from the manifest formal self-adjointness on $L^2(\R^N,dt)$ of the PDOs~$\cH_k$ given by~\eqref{cHk}--\eqref{uhr} with $\lambda>0$), we transfer their action to the factor $F_{N-1}$, noting that the argument $-s$ should be replaced by $s$, since there is no complex conjugation in \eqref{FN}. 

At this stage we can make use of our assumption \eqref{recass}, thus arriving at
\be\label{eForm}
:\hat{\Sigma}_k^{(N)}(\cL+\cE)(t):F_N(t,u) = \sum_{l=0}^ku_N^l\binom{N-k+l}{l}S_{k-l}^{(N-1)}(u_1-u_N,\ldots,u_{N-1}-u_N)F_N(t,u).
\ee
Invoking the recurrence \eqref{Srec} for the elementary symmetric functions, we can rewrite this eigenvalue formula as
\be\label{efk}
:\hat{\Sigma}_k^{(N)}(\cL+\cE)(t):F_N(t,u) = S_k^{(N)}(u)F_N(t,u),\quad k=1,\ldots,N.
\ee
Comparing this to~\eqref{recass}, we find that we have established a recursive procedure to construct joint eigenfunctions for any number of particles $N$. Indeed, we can start the recursion at $N=1$ with the plane wave
\be\label{F1}
F_1(t,u)\equiv \exp(itu),
\ee
which clearly satisfies
\be\label{1eig}
:\hat{\Sigma}_1^{(1)}(\cL+\cE)(t):F_1(t,u) = -i\partial_tF_1(t,u) = uF_1(t,u).
\ee
Proceeding by induction on $N$, it is now straightforward to verify that this recursive procedure yields the integral representation \eqref{intFN1} of $F_N$.

\section{The free case $\lambda =1$}
In this section we begin to back up the formal manipulations in the previous section by rigorous analysis. For the special choice $\lambda =1$ this turns out to be quite easy, inasmuch as we need only make some rather obvious changes in equations pertaining to the two free cases $b=a_{\pm}$ in our previous paper~\cite{HR13}. We shall cite these equations by using a prefix~I.

Specifically, choosing
\be\label{apam}
a_{-}=\pi,\ \ \ \ a_{+}=2, 
\ee
and replacing the variables $x,y$ of~{\it loc.~cit.} by~$t,u$, the function I(2.19) turns into the plane wave $F_1(t,u)$~\eqref{F1}.  Moreover, the kernel function I(3.1) and weight function I(3.2) amount to the kernel function (cf.~\eqref{cKs})
\be\label{K1}
\cK_N^{\sharp}(1;t,s)=\prod_{j=1}^N\prod_{k=1}^{N-1}\frac{1}{2\cosh(t_j-s_k)},
\ee
 and weight function (cf.~\eqref{cWW})
 \be\label{W1}
W_N(1;t)=\prod_{1\le j< k\le N}4\sinh^2(t_j-t_k).
\ee
 Finally, the recurrence I(2.16) becomes the recurrence~\eqref{FN}. 

As shown in~{\it loc.~cit.}, the relevant integrals can be recursively evaluated explicitly. Thus we need only use the above reparametrizations to deduce the following theorem from Theorem~3.1 in~\cite{HR13}.

\begin{theorem}
We have
\be\label{FNfree}
F_N(1;t,u)=\prod_{1\leq j<k\leq N}
\frac{-i\pi}{4\sinh(t_j-t_k)\sinh(\pi (u_j-u_k)/2)}\cdot 
\sum_{\sigma\in S_N}(-)^\sigma\exp\big(i t\cdot\sigma(u)\big),
\ee
where $t,u\in\R^N$, and $t_j\ne t_k, u_j\ne u_k$ for $j\ne k$.
\end{theorem}

Multiplying~\eqref{FNfree} by $W_N(1;t)^{1/2}$ and reverting to the positions $x$ and momenta $p$, we obtain (cf.~\eqref{dLess}--\eqref{cWW})
\be
\Psi_N(\hbar;x,p)= \prod_{1\leq j<k\leq N}
\frac{-i\pi}{2\sinh(\pi (p_j-p_k)/\hbar\mu)}\cdot
\sum_{\sigma\in S_N}(-)^\sigma\exp\big(i x\cdot\sigma(p)/\hbar\big).
\ee
Omitting the $p$-dependent product, we therefore obtain the kernel of the multivariate sine transform, in accord with~\eqref{PsiEqs} for~$g=\hbar$.

\section{The step from $N=1$ to $N=2$}\label{Sec3}
In this section we consider the $N=2$ case. Until further notice, we choose $u\in\R^2$. We also need to restrict the complex coupling~$\lambda$. To this end we introduce
\be\label{Lam}
\Lambda_{\kappa}\equiv \{ \lambda\in\C\mid \re \lambda >\kappa\},\ \ \ \kappa\ge 0.
\ee
In this section and the next one, we take $\lambda\in\Lambda_0$, but for the case $N>3$ treated in Section~6 we restrict attention to $\Lambda_1$ and its closure $\overline{\Lambda_1}$.  

The integrand arising in the first step of the recursive scheme is given by (cf.~\eqref{FN})
\be\label{I2}
\begin{split}
I_2(\lambda;t,u,s) &\equiv \cK_2^\sharp(\lambda;t,s)F_1(s,u_1-u_2)\\ &= \exp\big(is(u_1-u_2)\big)\prod_{j=1}^2\big[2\cosh(t_j-s)\big]^{-\lambda}.
\end{split}
\ee
In the $s$-plane it has upward/downward sequences of singularities located at
\be\label{I2sing}
s=t_j+\frac{i\pi}{2}(2n+1),\quad s=t_j-\frac{i\pi}{2}(2n+1),\quad j=1,2,\quad n\in\N.
\ee
Its asymptotic behavior for~$|\re s|\to\infty$ readily follows from the elementary estimate
\be\label{chest}
|\cosh(z)^{-\lambda}|<C(\lambda,|\im z|) \exp(-\re\lambda |\re z|), \ \ \ |\im z|<\pi/2,\ \ \ \lambda\in\Lambda_0,
\ee
where~$C$ is a continuous function on~$\Lambda_0\times [0,\pi/2)$.
Specifically, it satisfies
\be\label{I2as}
I_2(\lambda;t,u,s) = O(\exp(-2\re \lambda\,|\re s|)),\quad |\re s|\to\infty,\ \ \ \ u\in\R^2,
\ee
where the implied constant is uniform for $( \lambda,t,\im s)$ varying over compact subsets of $\Lambda_0\times\C^2\times\R$. 

Requiring at first $t\in\R^2$, the $s$-contour $\R$ stays below/above the upward/downward sequences~\eqref{I2sing}. Thus the function
\be\label{F2}
F_2(\lambda;t,u)\equiv \exp(iu_2(t_1+t_2))\int_{\R}ds I_2(\lambda;t,u,s),\quad \lambda>0,\quad t,u\in\R^2,
\ee
is well defined. 
Taking the uniform bound~\eqref{I2as} into account, we now infer that~$F_2(\lambda;t,u)$ continues analytically to~$\lambda$ in the right half plane~$\Lambda_0$ and $t\in\C^2$ satisfying $|\im t_j|<\pi/2$, $j=1,2$. (Indeed, by virtue of Cauchy's integral formula the~$t_j$-partials of the integrand satisfy the same bound, whereas for~$\lambda$-partials it suffices to replace the constant~2 in the exponent by any~$\rho<2$. From this the asserted holomorphy follows; moreover, it implies we can take partials of $F$ through the integral sign.)

The bound~\eqref{I2as} also entails we may shift the contour $\R$ up and down as long as we do not encounter any of the singularities. In particular, with $t\in\R^2$, we can shift the contour $\R$ in \eqref{F2} down to $\R-i\pi/2+i\epsilon$ for any $\epsilon\in(0,\pi/2)$. From the resulting integral representation we then see that $F_2$ extends holomorphically to $\im t_j\in (-\pi+\epsilon,\epsilon)$, $j=1,2$. This argument can be iterated, and we can likewise move the contour up step by step. In this way we can reach any strip of width $\pi$. More precisely, introducing the domain
\be
\cA_2\equiv \{t\in\C^2 \mid  |\im(t_1-t_2)|<\pi\},
\ee
we have the following result.

\begin{proposition}\label{h2Prop}
Let $u\in\R^2$. Then $F_2(\lambda;t,u)$ is holomorphic for $(\lambda,t)\in\Lambda_0\times \cA_2$.
\end{proposition}

\begin{proof}
Fixing $(\lambda,t)\in \Lambda_0\times \cA_2$, we introduce
\be
\eta\equiv \im(t_1+t_2)/2.
\ee
Then we have
\be
\im t_j\pm \pi/2\gtrless \eta,\quad j=1,2,
\ee
so the contour $\R+i\eta$ stays below/above the upward/downward sequences of singularities \eqref{I2sing}. Hence we can arrive at these $t$-values by successive shifts of the contour, as described above, without encountering any singularities. This yields holomorphy in a suitably restricted neighborhood of $(\lambda,t)$, and thus in $\Lambda_0\times \cA_2$.
\end{proof}

We continue to show that $F_2$ is a joint eigenfunction of the $N=2$ PDOs with the expected eigenvalues.  

\begin{proposition}\label{e2Prop}
Let $u\in\R^2$. For all $(\lambda,t)\in \Lambda_0\times \cA_2$, we have the joint eigenfunction property
\be\label{e2}
:\hat{\Sigma}_k^{(2)}(\cL+\cE)(t):F_2(t,u) = S_k^{(2)}(u)F_2(t,u),\quad k=1,2.
\ee
\end{proposition}

\begin{proof}
By virtue of the analyticity features obtained above, we need only show this for real~$t$ and $\lambda>0$ (say).  Following the argument leading from~\eqref{Sexp} to~\eqref{efk}, we deduce that it remains to show
\be\label{intEq}
\int_{\R}ds F_1(s,u_1-u_2):\hat{\Sigma}_k^{(2)}(\cL+\cE)(t):\cK_2^\sharp(t,s)= S_k^{(1)}(u_1-u_2)\int_{\R}ds I_2(t,u,s),\quad k=1,2.
\ee

For $k=2$, the eigenfunction identity \eqref{key1} implies that the integral vanishes. On the other hand, setting $k=1$ and making use of the kernel identity \eqref{key2}, we obtain
\be
\int_{\R}ds F_1(s,u_1-u_2):\hat{\Sigma}_1^{(1)}(\cL+\cE)(-s):\cK_2^\sharp(t,s).
\ee
Now the operator $:\hat{\Sigma}_1^{(1)}(\cL+\cE)(-s):$ reduces to $id/ds$, so we can transfer its action to the factor $F_1$ via integration by parts. (Note that the argument $-s$ is then replaced by $s$, since there is no complex conjugation involved.) Finally, we invoke the eigenvalue equation \eqref{1eig}, and arrive at the rhs of \eqref{intEq}.
\end{proof}

Next, we detail the relation between $F_2$ and the so-called conical (or Mehler) function, see also Section~4.2 in~\cite{Rui11}. We recall that the latter function can be defined by
\be
P_{ik-1/2}^{1/2-\lambda}(\cosh r)\equiv \frac{(\sinh r)^{\lambda-1/2}}{2^{\lambda-1/2}\Gamma(\lambda+1/2)}{}_2F_1(\lambda-ik,\lambda+ik;\lambda+1/2;(1-\cosh r)/2),
\ee
cf.~Eq.~14.3.15~in \cite{Dig10}. Performing the substitution $s\to s+(t_1+t_2)/2$ in the integral \eqref{F2}, we obtain
\be\label{F2s}
F_2(\lambda;(t_1,t_2),(u_1,u_2)) = \exp(i(t_1+t_2)(u_1+u_2)/2)F(\lambda;t_1-t_2,u_1-u_2),
\ee
with
\be\label{F2r}
F(\lambda;z,w)\equiv \int_{\R}ds\frac{\exp(isw)}{\prod_{\de=+,-}[2\cosh(s+\de z/2)]^\lambda}.
\ee
From the latter equation it is plain that $F$ is even in both $z$ and $w$, so that $F_2$ is invariant under the interchanges $t_1\leftrightarrow t_2$ and $u_1\leftrightarrow u_2$. Note that the former invariance property is clear from the defining formula~\eqref{F2}, whereas the latter is not immediate from~\eqref{F2}. Moreover, comparing \eqref{F2r} to Eq.~14.12.4 in \cite{Dig10}, we find
\be
F(\lambda;z,w) = \left(\frac{\pi}{4}\right)^{1/2}\frac{\Gamma(\lambda+iw/2)\Gamma(\lambda-iw/2)}{\Gamma(\lambda)(2\sinh z)^{\lambda-1/2}}P_{iw/2-1/2}^{1/2-\lambda}(\cosh z).
\ee

We proceed to establish a bound on the $2$-particle eigenfunction $F_2$, which exhibits its exponential decay as $|\re(t_1-t_2)|\to\infty$. This bound will be one of the key ingredients for achieving analytic control of the second step in the recursive scheme.

\begin{proposition}\label{F2bProp}
Let $u\in\R^2$. For any $(\lambda,t)\in \Lambda_0\times \cA_2$, we have
\be\label{F2b}
|F_2(\lambda;t,u)| < C(\lambda,|\im(t_1-t_2)|)\exp(-\im(t_1+t_2)(u_1+u_2)/2)\frac{\re(t_1-t_2)}{\sinh\big(\re \lambda\,\re(t_1-t_2)\big)},
\ee
where $C$ is continuous on $\Lambda_0\times [0,\pi)$.
\end{proposition}

\begin{proof}
In view of \eqref{F2s}, this amounts to a bound on the function $F(\lambda;z,w)$. Clearly, its representation \eqref{F2r} can be continued analytically to the strip $|\im z|<\pi$, and for $z$ in this strip we have
\be\label{Fb}
|F(\lambda;z,w)| \le  C(\lambda,|\im z|) \int_{\R}\frac{ds}{\prod_{\de=+,-}|2\cosh(s+\de z/2)|^{\re \lambda}},
\ee
where $C$ is a continuous  function on $ \Lambda_0\times [0,\pi)$.
In order to bound the rhs we note the inequality
\be\label{cIneq}
\frac{2\cosh (a\,\re v)}{|2\cosh v|^a}< C(a,|\im v|) , \ \ a >0,\ \ |\im v|<\pi/2,\ \ \ \forall\, \re v\in\R,
\ee
where $C$ is continuous on $(0,\infty)\times [0,\pi/2) $. 
Combining \eqref{Fb} and \eqref{cIneq}, we deduce
\be\label{Fbf}
|F(\lambda;z,w)| < C(\lambda,|\im z|)\int_{\R}\frac{ds}{\prod_{\de=+,-}2\cosh\big(\re \lambda\,(s+\de \re z/2)\big)},
\ee
with $C$ continuous on $ \Lambda_0\times [0,\pi)$. By a residue calculation it is straightforward to verify that the integral is given by $\re z/2\sinh(\re \lambda\, \re z)$. Clearly, this implies \eqref{F2b}.
\end{proof}

We have thus far viewed $u$ as a fixed vector in $\R^2$. 
From~\eqref{F2s}--\eqref{F2r} and the bound~\eqref{chest}, however, it easily follows that $F_2$ can be analytically continued to any $u\in\C^2$ with $|\im (u_1-u_2)|<2\re \lambda$. We proceed to estimate $F_2$ in the corresponding holomorphy domain.
 
\begin{proposition}
The function~$F_2(\lambda;t,u)$ is holomorphic in
\be\label{cD2}
\cD_2\equiv \{ (\lambda,t,u)\in \Lambda_0\times \cA_2\times \C^2 \mid  |\im (u_1-u_2)|<2\re \lambda \},
\ee
 and for
  $\im(u_1-u_2)\ne 0$ we have
\bea\label{F2bu}
|F_2(\lambda;t,u)| & <  &C(\lambda,|\im (t_1-t_2)|)\exp(- \im [(t_1+t_2)(u_1+u_2)]/2)
\nonumber \\
&  & \times \frac{\sinh( \im(u_2-u_1)\re(t_1-t_2)/2)}{\sin (\pi\im(u_2-u_1)/2\re \lambda)\sinh(\re \lambda\,\re(t_1-t_2))},
\eea
where   $C$ is continuous on $\Lambda_0\times[0,\pi)$. 
\end{proposition}
\begin{proof}
Just as in the previous proposition, this amounts to a bound on the function $F(\lambda;z,w)$. From~\eqref{F2r} we obtain as a generalization of~\eqref{Fb}
\be\label{Fb2}
|F(\lambda;z,w)| \le C(\lambda,|\im z|)  \int_{\R}ds\,\frac{\exp(-s\im w)}{\prod_{\de=+,-}|2\cosh(s+\de z/2)|^{\re \lambda}}.
\ee
 The counterpart of
\eqref{Fbf} is then
 \be\label{Fbfu}
|F(\lambda;z,w)| < C(\lambda,|\im z|)\int_{\R}ds\,\frac{\exp(-s\im w)}{\prod_{\de=+,-}2\cosh\big(\re \lambda\,(s+\de \re z/2)\big)}.
\ee
A straightforward residue calculation now yields
\be
\int_{\R}ds\,\frac{\exp(-s\im w)}{\prod_{\de=+,-}\cosh\big(a\,(s+\de \re z/2)\big)}=
\frac{2\pi \sinh(\im w\re z/2)}{a\sin (\pi\im w/2a)\sinh(a\,\re z)}, 
 \ee
where $a>0, |\im w|\in (0,2a)$,  and so~\eqref{F2bu} follows.
\end{proof}

Prop.~4.3 can be obtained from Prop.~4.4 by letting $\im (u_1-u_2)$ converge to 0, but a separate treatment of real and complex~$u$ is expedient for later purposes.   Specifically, Prop.~4.3 is used in the next section to obtain the $N=3$ counterparts of Props.~4.1 and 4.2, dealing with the joint eigenfunction properties for real $u\in\R^3$, and then we obtain the $N=3$ analogs of Props.~4.3 and 4.4. Finally, this flow chart can be used for the inductive step taken in Section~6. (Of course, once the joint eigenvalue equations are proved for real~$u$, they continue to the pertinent complex~$u$. Indeed, the eigenvalues~$S_k(u)$ are entire functions of~$u$.)

\section{The step from $N=2$ to $N=3$}\label{Sec4}
In this section we consider the analytic aspects arising in the step $N=2\to N=3$ of the recursive scheme. As far as possible, we shall follow the discussion in Section \ref{Sec3}.  

We recall that the integrand in question is given by (cf.~\eqref{FN})
\be
I_3(\lambda;t,u,s)\equiv W_2(\lambda;s)\cK_3^\sharp(\lambda;t,s)F_2(\lambda;s,(u_1-u_3,u_2-u_3)).
\ee
From Prop.~\ref{F2bProp} we can deduce a suitable bound on the factor $F_2$. We shall consider $s\in\C^2$ with $\im(s_1-s_2)=0$, anticipating simultaneous shifts of the integration contours. Letting $u\in\R^3$, we obtain from \eqref{F2b}
\be\label{F2b2}
\begin{split}
|F_2(\lambda;s,(u_1-u_3,u_2-u_3))| &< C(\lambda)\exp(-\im(s_1+s_2)(u_1+u_2-2u_3)/2)\\ &\quad \times \frac{s_1-s_2}{\sinh(\re \lambda(s_1-s_2))},\quad \im s_1=\im s_2,
\end{split}
\ee
with $C$ continuous on $\Lambda_0$. Recalling the definition~\eqref{cWW} of~$W_2$, we easily deduce the estimate
\begin{multline}\label{W2F2}
|W_2(\lambda;s)F_2(\lambda;s,(u_1-u_3,u_2-u_3))| < C(\lambda)\exp(-c(u_1+u_2-2u_3))\\ 
\quad \times \prod_{j=1}^2(1+|\re s_j|)\exp(\re \lambda |\re s_j|),\quad c:=\im s_1=\im s_2.
\end{multline}
Finally, taking into account the factor~$\cK_3^\sharp$~\eqref{cKs}, we deduce from the bound~\eqref{chest} that we have
\begin{multline}\label{I3b}
|I_3(\lambda;t,u,s)|<C(\lambda,\re t,|\im t_1-c|,|\im t_2-c|,|\im t_3-c|)\\
\times\prod_{j=1}^2(1+|\re s_j|)\exp(-2\re \lambda |\re s_j|), \ \ \ u\in\R^3,
\end{multline}
with $C$ continuous on $\Lambda_0\times\R^3\times [0,\pi/2)^3$. 

To begin with, we assume $t\in\R^3$. This entails that the contour $\R$ in the $s_k$-plane stays away from the singularities at
\be\label{I3sing}
s_k = t_j \pm \frac{i\pi}{2}(2n+1),\quad k=1,2,\quad j=1,2,3,\quad n\in\N,
\ee
so that the function
\be\label{F3}
F_3(\lambda;t,u)\equiv \frac{1}{2}\exp(iu_3(t_1+t_2+t_3))\int_{\R^2}ds I_3(\lambda;t,u,s),\quad  \lambda>0,\quad t,u\in\R^3,
\ee
is well defined.

The uniform bound~\eqref{I3b}  implies that   $F_3$  extends to a holomorphic function of~$\lambda$ in~$\Lambda_0$ and of~$t$ for~$|\im t_j|<\pi/2$, $j=1,2,3$. Moreover, we can shift the two contours $\R$ simultaneously as long as we do not meet any of the singularities \eqref{I3sing}. Using the same iterative procedure as in the $N=2$ case, we can thus extend the holomorphy domain step by step. To detail this, we introduce the domain
\be\label{cA3}
\cA_3\equiv \{t\in\C^3 \mid   \max_{1\leq j<k\leq 3}|\im(t_j-t_k)|<\pi\}.
\ee
 Also, given $t\in\C^3$, we set
\be\label{phi}
\phi(t)\equiv \im(t_{j_1}+t_{j_3})/2,
\ee
where the indices $j_1$ and $j_3$ are determined by the requirement
\be\label{imord}
\im t_{j_3}\leq \im t_{j_2}\leq \im t_{j_1},\quad \{j_1,j_2,j_3\} = \{1,2,3\}.
\ee
Then we have the following counterpart of Prop.~\ref{h2Prop}.

\begin{proposition}\label{h3Prop}
Let $u\in\R^3$. Then $F_3(\lambda;t,u)$ is holomorphic in $\Lambda_0\times \cA_3$.
\end{proposition}

\begin{proof}
We fix $(\lambda,t)\in \Lambda_0\times \cA_3$, and note
\be
\im t_j\pm \pi/2\gtrless \phi(t),\quad j=1,2,3.
\ee
Hence the contour $\R+i\phi(t)$ stays below/above the upward/downward sequences of singularities \eqref{I3sing}. This entails we can extend $F_3$ holomorphically to any such $t$-value by simultaneous contour shifts, without meeting any singularities. More specifically, in this way we arrive at the representation
\be\label{F3phi}
F_3(\lambda;t,u) =\frac{1}{2} \exp(iu_3(t_1+t_2+t_3))\int_{(\R+i\phi(t))^2}ds I_3(\lambda;t,u,s), \quad (\lambda,t)\in \Lambda_0\times \cA_3,\quad u\in\R^3,
\ee
with the uniform bound~\eqref{I3b} ensuring holomorphy in $\Lambda_0\times \cA_3$.
\end{proof}

We proceed to show that $F_3$ is a joint eigenfunction of the pertinent PDOs.

\begin{proposition}\label{e3Prop}
Let $u\in\R^3$.  For all $(\lambda,t)\in \Lambda_0\times \cA_3$, we have the joint eigenfunction property
\be
:\hat{\Sigma}_k^{(3)}(\cL+\cE)(t):F_3(t,u) = S_k^{(3)}(u)F_3(t,u),\quad k=1,2,3.
\ee
\end{proposition}

\begin{proof}
In view of Prop.~5.1  it suffices to show this for $t\in\R^3$ and $\lambda >2$ (say).  Reasoning just as in the proof of Prop.~4.2, we deduce that  it remains to show that the integral in \eqref{F3} has the joint eigenfunction property
\begin{multline}\label{intEq2}
\int_{\R^2}ds W_2(s)F_2(s,(u_1-u_3,u_2-u_3)):\hat{\Sigma}_k^{(3)}(\cL+\cE)(t):\cK_3^\sharp(t,s)\\
 = S_k^{(2)}(u_1-u_3,u_2-u_3)\int_{\R^2}ds I_3(t,u,s),\quad k=1,2,3,
\end{multline}
with $S_3^{(2)}\equiv 0$.  
For $k=3$, the integral on the lhs vanishes (cf.~\eqref{key1}), and in the remaining two cases \eqref{key2} yields
\be\label{2int}
\int_{\R^2}ds W_2(s)F_2(s,(u_1-u_3,u_2-u_3)):\hat{\Sigma}_k^{(2)}(\cL+\cE)(-s):\cK_3^\sharp(t,s),\quad k=1,2.
\ee
Using~\eqref{sim2}, we can rewrite this as
\be\label{intcHr}
\int_{\R^2}ds W_2(s)^{1/2}F_2(s,(u_1-u_3,u_2-u_3)) \cH_k^{(2)}(-s) W_2(s)^{1/2}\cK_3^\sharp(t,s),\quad k=1,2.
\ee
Since we choose $\lambda>2$, the function~$W_2(s)^{1/2}$ is in $C^2(\R^2)$, whereas $F_2$ and~$\cK_3^{\sharp}$ are smooth in~$s$. Recalling~\eqref{cHk}, we see that the factor~$W_2(s)^{1/2}$ is differentiated at most twice. Thus we can integrate by parts to get
\be\label{intcHl}
\int_{\R^2}ds W_2(s)^{1/2}\cK_3^\sharp(t,s)\cH_k^{(2)}(s)W_2(s)^{1/2}F_2(s,(u_1-u_3,u_2-u_3)) ,\quad k=1,2.
\ee

The upshot is that \eqref{2int} is given by
\be\label{upsh}
\int_{\R^2}ds W_2(s)\cK_3^\sharp(t,s):\hat{\Sigma}_k^{(2)}(\cL+\cE)(s):F_2(s,(u_1-u_3,u_2-u_3)),\quad k=1,2.
\ee
Thus we need only appeal to the eigenvalue properties \eqref{e2} to arrive at the rhs of \eqref{intEq2}.  
\end{proof}

We continue to establish the $N=3$ analog of Prop.~\ref{F2bProp}. As we shall see, this can be reduced to the calculation of the integral
\be\label{B2}
B_2(w)\equiv \int_{\R^2}dv \frac{(v_1-v_2)\sinh(v_1-v_2)}{\prod_{j=1}^3\prod_{k=1}^2\cosh(w_j-v_k)},
\ee
which occurs in a suitable majorization of $F_3$. We met this integral in our paper~\cite{HR13}, where it played a similar role in bounding joint eigenfunctions. Moreover, we showed that $B_2$ is given by
\be\label{B2eval}
B_2(w) = 4\prod_{1\leq j<k\leq 3}\frac{w_j-w_k}{\sinh(w_j-w_k)},
\ee
cf.~Lemma C.2 in \cite{HR13} for $N=2$.

To formulate the precise result we recall the definition \eqref{phi} of $\phi(t)$ and introduce in addition
\be\label{d}
d(t)\equiv \im(t_{j_1}-t_{j_3}),
\ee
where the indices $j_1$ and $j_3$ are again given by \eqref{imord}.

\begin{proposition}\label{F3bProp}
Let $u\in\R^3$. For any   $(\lambda,t)\in \Lambda_0\times \cA_3$, we have
\be\label{F3b}
\begin{split}
|F_3(\lambda;t,u)| &< C(\lambda,d(t))\exp(-(u_1+u_2)\phi(t)-u_3\im t_{j_2})\\ &\quad \times \prod_{1\leq j<k\leq 3}\frac{\re(t_j-t_k)}{\sinh\big(\re \lambda\,\re(t_j-t_k)\big)},
\end{split}
\ee
where $C$ is continuous on $\Lambda_0\times [0,\pi)$.
\end{proposition}

\begin{proof}
From the representation \eqref{F3phi} of $F_3$ and the definitions \eqref{cWW} of $W_2$ and \eqref{cKs} of $\cK_3^\sharp$, we obtain the majorization
\begin{multline}
|F_3(\lambda;t,u)|\leq \exp(-u_3\im(t_1+t_2+t_3))\\ \times\frac12\int_{(\R+i\phi(t))^2}ds \frac{[4\sinh^2(s_1-s_2)]^{\re\lambda}}{\prod_{j=1}^3\prod_{k=1}^2|2\cosh(t_j-s_k)|^{\re \lambda}}|F_2(\lambda;s,(u_1-u_3,u_2-u_3))|.
\end{multline}
Using next the bounds~\eqref{cIneq} and~\eqref{F2b2}, this yields
\begin{multline}
|F_3(\lambda;t,u)|\leq C(\lambda,d(t))\exp\big(-u_3\im t_{j_2}-(u_1+u_2)\phi(t)\big) \\ \times\int_{(\R+i\phi(t))^2}ds \frac{\re(s_1-s_2)[\sinh^2(s_1-s_2)]^{\re\lambda} }{\sinh(\re\lambda\,\re(s_1-s_2))}\prod_{j=1}^3\prod_{k=1}^2\frac{C^\prime(\lambda,|\im t_j-\phi(t)|)}{\cosh(\re\lambda\,\re(t_j-s_k))},
\end{multline}
where $C^\prime$ is continuous on $ \Lambda_0\times[0,\pi/2)$. Now we have from \eqref{phi}, \eqref{d} and \eqref{cA3},
\be
|\im t_j-\phi(t)|\leq d(t)/2 < \pi/2,\quad j=1,2,3,\quad t\in \cA_3.
\ee
It follows that the $C^\prime$-product is bounded above by a function $C(\lambda,d(t))$ that is continuous on $ \Lambda_0\times[0,\pi)$.
Thus, changing variables, we finally arrive at
\begin{multline}\label{F3fin}
|F_3(\lambda;t,u)|\leq C(\lambda,d(t))\exp(-u_3\im t_{j_2}-(u_1+u_2)\phi(t)) \\ \times\int_{\R^2}dr \frac{[\sinh^2(r_1-r_2)]^{\re\lambda} }{\sinh(\re\lambda\,(r_1-r_2))}\frac{r_1-r_2}{\prod_{j=1}^3\prod_{k=1}^2\cosh(\re\lambda\,(\re t_j-r_k))}.
\end{multline}

We proceed to obtain a bound on the sinh-ratio factor in the integrand, considering first the case~$ \lambda\in\overline{\Lambda_{ 1}}$, cf.~\eqref{Lam}. To this end we make use of \eqref{cIneq} and the inequality
\be\label{sIneq}
(\sinh z)^a\leq \sinh a z,\quad z\geq 0,\quad a\geq 1,
\ee
which can be verified as follows. Introducing the function $\varphi(a,z)\equiv \sinh (a z)/(\sinh z)^a$ for $a\geq 1$ and $z>0$, we have
\be
\frac{d\varphi(a,z)}{da} = \varphi(a,z)\big(z\coth a z - \log(\sinh z)\big).
\ee
Since $\coth a z\geq 1$ and $\log(\sinh z)\leq z-\log 2$, the rhs is positive. Given that $\varphi(1,z)\equiv 1$, it follows that $\varphi(a,z)\geq 1$, which in turn implies \eqref{sIneq}.

Combining \eqref{sIneq} with the bound \eqref{F3fin}, we infer
\begin{multline}\label{F3f1}
|F_3(\lambda;t,u)|\leq C(\lambda,d(t))\exp(-u_3\im t_{j_2}-(u_1+u_2)\phi(t)) \\ \times\int_{\R^2}
dr \frac{(r_1-r_2)\sinh(\re\lambda(r_1-r_2))}{\prod_{j=1}^3\prod_{k=1}^2\cosh(\re\lambda(\re t_j-r_k))}.
\end{multline}
When we now take $r\to v/\re\lambda$ in the  integral 
and compare the result to \eqref{B2}, then the bound \eqref{F3b} with~$\re\lambda\ge 1$  follows from \eqref{B2eval}.

It remains to extend this estimate to the interval~$\re\lambda\in(0,1)$. To this end we set $\rho:=\sinh^{-1}(1)$ and employ the bound
\be\label{sb1}
\frac{r(\sinh^2r)^{\re\lambda}}{\sinh(r\re\lambda)}\le \frac{1}{\re\lambda},\ \ \ r\in [-\rho,\rho],\ \ \ \re\lambda>0, 
\ee
which is clear from the positive function $x/\sinh x $, $x\in\R$, being bounded above by~1.
We also use
\be\label{sb2}
(\sinh^2r)^{\re\lambda}\le C(\re\lambda)\sinh^2(r\re\lambda),\ \ \ |r|\ge \rho,
\ee
with $C$ continuous on $(0,\infty)$. We now split the integral over~$\R^2$ in~\eqref{F3fin} into an integral over the set~$|r_1-r_2|>\rho$ and the integral over the complement~$|r_1-r_2|\le\rho$. For the first integral we can invoke the bound~\eqref{sb2}. Doing so, we can extend the integration to~$\R^2$ with the same integrand. This yields the integral in~\eqref{F3f1} already handled, so the first integral is majorized by the rhs of~\eqref{F3b}, now with $C$ continuous on~$\Lambda_0\times [0,\pi)$.

Finally, we consider the second integral. Here we can use~\eqref{sb1}, and then extend the integration over~$\R^2$ to obtain the square of the integral~$\cI_3(\re\lambda\,\re t)$, with
\be\label{chint}
\cI_3(q)\equiv \int_{\R}\frac{dr}{\prod_{j=1}^3\cosh(r- q_j)},\ \ \ q\in\R^3.
\ee
This integral can be evaluated by a contour integration, the result being
\be
\cI_3(q)=\pi/2\prod_{1\le j<k\le 3}\cosh((q_j-q_k)/2).
\ee
Now it is easy to check the estimate
\be
\frac{1}{\prod_{1\le j<k\le 3}\cosh^2((q_j-q_k)/2)}\le \prod_{1\le j<k\le 3}\frac{q_j-q_k}{\sinh(q_j-q_k)}.
\ee
Therefore, the second integral is also majorized by the rhs of~\eqref{F3b}, completing the proof.
\end{proof}

We proceed to obtain a counterpart of Prop.~4.4. To this end we introduce
\be\label{not3}
T_3\equiv \frac13\sum_{j=1}^3t_j,\ \ U_3\equiv \frac13\sum_{j=1}^3u_j,\ \ \tilde{t}_j\equiv  t_j-T_3,\ \ \ 
\tilde{u}_j\equiv  u_j-U_3,\ \ \ j=1,2,3,
\ee
and recall the representation~\eqref{F3phi}. Substituting~\eqref{F2s}, it is not hard to verify that it implies
\be\label{F3s}
F_3(\lambda;t,u)=\exp(3iT_3U_3)F_3^r(\lambda;t,u),
\ee
\be\label{F3r}
F_3^r(\lambda;t,u)\equiv\frac{1}{2}\int_{\R^2}dsW_2(\lambda;s)\cK^{\sharp}_3(\lambda;\tilde{t},s)F_2(\lambda;s,(u_1-u_3,u_2-u_3)),\ \  (\lambda,t,u)\in\Lambda_0\times \cA_3^r\times \R^3,
\ee
where~$\cA_3^r$ is the subset of~$\cA_3$~\eqref{cA3} given by 
\be\label{cA3r}
A_3^r\equiv \{ t\in \C^3 \mid \mu(t)<\pi/2 \},
\ee
with $\mu(t)$ the maximum function
\be
\mu(t)\equiv \max_{j=1,2,3}|\im \tilde{t}_j|.
\ee

The representation~\eqref{F3s}--\eqref{F3r} is of interest in its own right, inasmuch as it explicitly shows that~$F_3(\lambda;t,u)$ is the product of a ``center-of-mass factor" and a function~$F_3^r(\lambda;t,u)$ whose dependence on~$t$ and~$u$ is encoded in the differences~$t_1-t_2,t_2-t_3$ and $u_1-u_2,u_2-u_3$. We use it as the starting point for the proof of the following analog of~Prop.~4.4.

Before stating the result, we would like to highlight a second key ingredient of its proof. This is the integral
\be\label{C2}
C_2(z,w)\equiv \int_{\R^2}dv\frac{\sinh(v_1-v_2)}{\prod_{j=1}^3\prod_{k=1}^2\cosh(w_j-v_k)}\exp\big(v_1(z_3-z_1)+v_2(z_3-z_2)\big),
\ee
where we need~$|\re(z_3-z_i)|<2$, $i=1,2$, and $|\im w_j|<\pi/2$, $j=1,2,3$, to get absolute convergence. When we assume in addition~$z_i\ne z_3$, $i=1,2$, and $w_j\ne w_k$, $1\le j<k\le 3$, it is given by
\be\label{C2ev}
C_2(z,w) = \prod_{j=1}^2\frac{\pi}{\sin (\pi (z_{3}-z_j)/2)}\prod_{1\leq j<k\leq 3}\frac{1}{\sinh(w_j-w_k)}\cdot\sum_{\tau\in S_{3}}(-)^\tau\exp\Big(\sum_{j=1}^2w_{\tau(j)}(z_{3}-z_j)\Big).
\ee
Just as the integral~$B_2(w)$ given by~\eqref{B2}--\eqref{B2eval} (to which~\eqref{C2}--\eqref{C2ev} reduce for $z=0$), it seems not easy to obtain this explicit evaluation in one fell swoop. It is a consequence of the recursive evaluation of more general integrals dealt with in Lemma~C.1 of~\cite{HR13}, and amounts to Lemma~C.3 for $N=2$. 

\begin{proposition}
The function~$F_3(\lambda;t,u)$ is holomorphic in
\be\label{cD3}
\cD_3\equiv \{ (\lambda,t,u)\in \Lambda_0\times \cA_3^r\times \C^3 \mid  |\im (u_j-u_k)|<2\re \lambda,\ \ j,k=1,2,3 \}.
\ee
 Moreover, for all~$(\lambda,t,u)\in\cD_3$ such that
  \be
\re\lambda\ge 1,\ \ \  \im(t_j-t_k)\ne 0,\ \ \   \im(u_j-u_k)\ne 0,\ \ \ 1\le j<k\le 3,
  \ee
   we have
\begin{multline}\label{F3bu}
|F_3(\lambda;t,u)|  <  C(\lambda,\mu(t))\exp(- 3\im (T_3U_3) )
 \\
 \times \frac{\sum_{\tau\in S_3}(-)^{\tau}\exp\Big(-\sum_{j=1}^3\re \big(\tilde{t}_{\tau(j)}\big)\im \tilde{u}_j\Big) }
{\prod_{1\le j<k\le 3}\sin (\pi\im(u_k-u_j)/2\re \lambda)\sinh(\re \lambda\,\re(t_j-t_k))},
\end{multline}
where   $C$ is continuous on $\overline{\Lambda_1}\times[0,\pi/2)$. 
\end{proposition}
\begin{proof}
We need only bound~$F_3^r$, cf.~\eqref{F3s}. Combining the bound~\eqref{cIneq} and Prop.~4.4, we first arrive at
\bea\label{F3bu1}
|F_3^r(\lambda;t,u)| & <  &\frac{C(\lambda,\mu(t))}{\sin(\pi \im (u_2-u_1)/2\re\lambda)}\int_{\R^2}ds\exp\big(- (s_1+s_2)\im (u_1+u_2-2u_3)/2\big)
\nonumber \\
  &  &  \times\frac{[\sinh^2(s_1-s_2)]^{\re\lambda}}{\sinh(\re\lambda\, (s_1-s_2))}\frac{\sinh(\im (u_2-u_1)(s_1-s_2)/2)}{\prod_{j=1}^3\prod_{k=1}^2\cosh(\re\lambda(\re \tilde{t}_j-s_k))}.
\eea
Choosing $\re\lambda\ge 1$ until further notice, we can invoke the bound~\eqref{sIneq} to deduce that we may replace the sinh-ratio in the integral by~$\sinh(\re\lambda\, (s_1-s_2))$. Then we switch to variables
\be\label{nvar}
v\equiv \re\lambda \cdot s,\ \ w\equiv   \re\lambda \cdot \re\tilde{t},\ \ \ z\equiv \im u/\re\lambda,
\ee
so that the resulting integral becomes equal to $(\re\lambda)^{-2}$ times
\be
 \int_{\R^2}dv\exp\big(-\frac12(v_1+v_2)(z_1+z_2-2z_3)\big)\frac{\sinh(v_1-v_2)\sinh\big(\frac12 (z_2-z_1)(v_1-v_2)\big)}
{\prod_{j=1}^3\prod_{k=1}^2\cosh(v_k-w_j)}.
\ee
In turn, this integral can be rewritten as
\be
\frac12  \int_{\R^2}dv\frac{\sinh(v_1-v_2)}{\prod_{j=1}^3\prod_{k=1}^2\cosh(v_k-w_j)}
\Big(\exp\big(v_1(z_3-z_1)+v_2(z_3-z_2)\big)-\big( v_1\leftrightarrow v_2\big)\Big).
\ee
Comparing to~\eqref{C2}, we deduce that the integral equals~$C_2(z,w)$. Next we use the evaluation~\eqref{C2ev} and reverse the substitution~\eqref{nvar}. Using the easily verified identity
\be\label{tuid3} 
\sum_{j=1}^2\re(\tilde{t}_{\tau(j)})\im(u_{3}-u_j)=-\sum_{j=1}^3 \re(\tilde{t}_{\tau(j)})\im \tilde{u}_j,
\ee
we now obtain the estimate~\eqref{F3bu}.

Turning to the general case $\re\lambda>0$, we split the integral in~\eqref{F3bu1} into integrals over the region $|s_1-s_2|>\rho$ and its complement. Using first~\eqref{sb2} in the first integral and then extending the integration over all of $\R^2$, we readily obtain once again a bound of the form occurring on the rhs of~\eqref{F3bu}. In the second integral we can first majorize the three factors depending on~$s_1-s_2$  by a function $C(\re\lambda)$ that is continuous on~$(0,\infty)$ and then extend the integration over $\R^2$. This results in the product of two one-dimensional integrals that are manifestly absolutely convergent on~$\cD_3$. Hence the holomorphy assertion follows, completing the proof.
\end{proof}

\section{The case $N>3$}
In this section we use Props.~\ref{h3Prop}--5.4  as the starting point for an induction argument.  More specifically, Thms.~\ref{hNThm}--6.4 follow from Props.~\ref{h3Prop}--5.4 for $N=3$, and our induction assumption is that the theorems are valid if we substitute $N-1$ for $N$.  

As will transpire, much of our discussion in Section \ref{Sec4} can be readily adapted to the general-$N$ case. Therefore we omit some details that will be clear from Section \ref{Sec4}.
On the other hand, we restrict attention to $\lambda\in \overline{\Lambda_1}$ (recall~\eqref{Lam}). Indeed, for the case $N=3$ it already became clear in the proof of Prop.~5.3 that the interval $\re\lambda \in(0,1)$ must be handled in a different way. This  supplementary method runs into novel difficulties for $N>3$, which we shall not address here. 

When constructing $F_N$ from $F_{N-1}$ we encounter the integrand (cf.~\eqref{FN})
\be
I_N(\lambda;t,u,s)\equiv W_{N-1}(\lambda;s)\cK_N^\sharp(\lambda;t,s)F_{N-1}(\lambda;s,(u_1-u_N,\ldots,u_{N-1}-u_N)),
\ee
where we choose $u=(u_1,\ldots,u_N)\in\R^N$. In order to bound the factor $F_{N-1}$, we assume that the imaginary parts of $s_1,\ldots,s_{N-1}$ are equal to some $c\in\R$, and replace $N$ by $N-1$ in Thm.~\ref{FNbThm} below. Under this assumption the difference function $d_{N-1}(s)$ (given by \eqref{dN} and \eqref{imordN}) vanishes and $\phi_{N-1}(s)$ (defined in \eqref{phiN} below) equals $c$. Hence we obtain the bound
\begin{multline}\label{FN-1b}
|F_{N-1}(\lambda;s,(u_1-u_N,\ldots,u_{N-1}-u_N))| < C(\lambda)\exp\Big(-c\sum_{j=1}^{N-1}(u_j-u_N)\Big)
\\
\prod_{1\leq m<n\leq N-1}\frac{s_m-s_n}{\sinh(\re\lambda(s_m-s_n))},\ \ \ c:=\im s_1=\cdots =\im s_{N-1},
\end{multline}
where $C$ is continuous on $ \overline{\Lambda_1}$. Combining this bound with \eqref{cWW} and \eqref{cKs} in the same way as in the previous section, we obtain the majorization
\begin{multline}\label{INb}
|I_N(\lambda;t,u,s)|<C(\lambda,\re t,|\im t_1-c|,\ldots,|\im t_N-c|)\\
\times\prod_{j=1}^{N-1}\big(1+|\re s_j|^{N-2}\big)\exp(-2\re \lambda |\re s_j|), \ \ \ u\in\R^N,
\end{multline}
with $C$ continuous on $\overline{\Lambda_1}\times\R^N\times [0,\pi/2)^N$.  

The singularities of the kernel function $\cK_N^\sharp(t,s)$ are located at
\be\label{INsing}
s_k = t_j\pm \frac{i\pi}{2}(2n+1),\quad k=1,\ldots,N-1,\quad j=1,\ldots,N,\quad n\in\N,
\ee
so when we choose at first $t\in\R^N$, then the function
\be\label{FNdef}
F_N(\lambda;t,u)\equiv \frac{\exp(iu_N(t_1+\cdots+t_N))}{(N-1)!}\int_{\R^{N-1}}ds I_N(\lambda;t,u,s),\quad \lambda\in\overline{\Lambda_1},\quad t,u\in\R^N,
\ee
is well defined. Furthermore,  $F_N$   extends to a holomorphic function of~$(\lambda,t)$ for~$\lambda\in\Lambda_1$ and  for $|\im t_j|<\pi/2$, $j=1,\ldots,N$. 

At this point we would like to mention that for the case $N=4$ we can still allow~$\lambda\in\Lambda_0$ in the above, since this is the restriction we have in the bound~\eqref{F3b} on~$F_3$. But we shall only obtain the $N>3$ counterpart of this bound for~$\lambda\in\overline{\Lambda_1}$, which is why we need to restrict~$\lambda$ for~$N>4$.

As before, we are allowed to shift all contours~$\R$ up and down by the same amount, provided the singularities \eqref{INsing} are not met, so we can extend the holomorphy domain step by step. To detail this, we let $t\in\C^N$ and introduce
 indices  
\be\label{imordN}
\im t_{j_N}\leq \im t_{j_{N-1}}\leq\cdots\leq \im t_{j_2}\leq \im t_{j_1},\quad \{j_1,\ldots,j_N\} = \{1,\ldots,N\},
\ee
and a function
\be\label{phiN}
\phi_N(t)\equiv \im(t_{j_1}+t_{j_N})/2.
\ee
We are now prepared for the first theorem of this section.
 
\begin{theorem}\label{hNThm}
Let $u\in\R^N$. Then the function $F_N(\lambda;t,u)$ is holomorphic for~$(\lambda,t)\in\Lambda_1\times \cA_N$, where
\be
\cA_N\equiv \{t\in  \C^N \mid \max_{1\leq j<k\leq N}|\im(t_j-t_k)|<\pi\}.
\ee
\end{theorem}

\begin{proof}
Fixing $t\in\C^N$, we have
\be
\im t_j\pm \pi/2\gtrless \phi_N(t),\quad j=1,\ldots,N.
\ee
Consequently, the $s_k$-contour $\R+i\phi_N(t)$ remains below/above the upward/downward sequences of singularities \eqref{INsing}. Hence, by simultaneous contour shifts, we can continue $F_N$, as given by~\eqref{FNdef}, to all~$t\in\cA_N$. In this way we arrive at the representation
\begin{multline}\label{FNphi}
F_N(\lambda;t,u) = \frac{\exp(iu_N(t_1+\cdots t_N))}{(N-1)!}\int_{(\R+i\phi_N(t))^{N-1}}ds I_N(\lambda;t,u,s),\\ \lambda\in\overline{\Lambda_1},\quad t\in \cA_N,\quad u\in\R^N.
\end{multline}
 Combined with the uniform bound~\eqref{INb}, this yields the asserted holomorphy properties.
\end{proof}

We continue to show that $F_N$ is a joint eigenfunction of the $N$ PDOs in question with the expected eigenvalues.  

\begin{theorem}\label{eNThm}
Let $u\in\R^N$. For all $(\lambda,t)\in \overline{\Lambda_1}\times \cA_N$, we have the joint eigenfunction property
\be\label{eNEq}
:\hat{\Sigma}_k^{(N)}(\cL+\cE)(t):F_N(t,u) = S_k^{(N)}(u)F_N(t,u),\quad k=1,\ldots,N.
\ee
\end{theorem}
\begin{proof}
Assuming this is the case for $N\to N-1$, it remains to establish
\begin{multline}\label{intEq3}
:\hat{\Sigma}_k^{(N)}(\cL+\cE)(t):\int_{\R^{N-1}}ds I_N(t,u,s)\\ = S_k^{(N-1)}(u_1-u_N,\ldots,u_{N-1}-u_N)\int_{\R^{N-1}}ds I_N(t,u,s),\quad k=1,\ldots,N,
\end{multline}
with $S_N^{(N-1)}\equiv 0$. Also, by the above we may restrict attention to real~$t$ and~$\lambda>2$ (say). (Indeed, by analyticity the result first follows for $(\lambda,t)\in\Lambda_1\times \cA_N$, and then continuity of the integral for~$\lambda\in\overline{\Lambda_1}$ yields the assertion.) Now by holomorphy in~$t$ we can take the differentiations under the integral sign and act with the PDOs on the kernel function. Making use of the eigenfunction identity \eqref{key1} and kernel identities \eqref{key2}, we find that the lhs of \eqref{intEq3} equals zero for $k=N$ and is given by
\be\label{SInt}
\int_{\R^{N-1}}ds W_{N-1}(s)F_{N-1}(s,(u_1-u_N,\ldots,u_{N-1}-u_N)):\hat{\Sigma}_k^{(N-1)}(\cL+\cE)(-s):\cK_N^\sharp(t,s),
\ee
for $k=1,\ldots,N-1$.

Next, we mimic the argument expressed in~\eqref{2int}--\eqref{upsh} (now involving the similarity-transformed Hamiltonians~$\cH^{(N-1)}_k$, cf.~\eqref{cHk}), which shows that~\eqref{SInt} equals
\be
\int_{\R^{N-1}}ds W_{N-1}(s)\cK_N^\sharp(t,s):\hat{\Sigma}_k^{(N-1)}(\cL+\cE)(s):F_{N-1}(s,(u_1-u_N,\ldots,u_{N-1}-u_N)).
\ee
Then we can make use of the eigenvalue equations \eqref{eNEq} for $N\to N-1$ to arrive at the rhs of \eqref{intEq3}. 
\end{proof}

We continue by deducing the announced uniform bound on $F_N$. We shall follow closely the proof of Prop.~\ref{F3bProp}, the key ingredient being the integral
\be\label{BN-1}
B_{N-1}(w)\equiv \int_{\R^{N-1}}dv\frac{\prod_{1\leq m<n\leq N-1}(v_m-v_n)\sinh(v_m-v_n)}{\prod_{j=1}^{N}\prod_{k=1}^{N-1}\cosh(w_j-v_k)}.
\ee
From Lemma C.2 in \cite{HR13} we recall its explicit evaluation
\be\label{BN-1eval}
B_{N-1}(w) = 2^{N-1}\prod_{1\leq m<n\leq N}\frac{w_m-w_n}{\sinh(w_m-w_n)}.
\ee
We shall also make use of the function $\phi_N(t)$ (cf.~\eqref{phiN}) and the distance function
\be\label{dN}
d_N(t)\equiv \im(t_{j_1}-t_{j_N}),\quad t\in\C^N,
\ee
where the indices $j_1$ and $j_N$ are given by \eqref{imordN}.

\begin{theorem}\label{FNbThm}
Let $u\in\R^N$. For any $(\lambda,t)\in \overline{\Lambda_1}\times \cA_N$, we have
\be\label{FNb}
\begin{split}
|F_N(\lambda;t,u)| &< C(\lambda,d_N(t))\exp\left(-\left[\phi_N(t)\sum_{j=1}^Nu_j+u_N\sum_{k=2}^{N-1}(\im t_{j_k}-\phi_N(t))\right]\right)\\ &\quad \times\prod_{1\leq m<n\leq N}\frac{\re(t_m-t_n)}{\sinh(\re \lambda\, \re(t_m-t_n))},
\end{split}
\ee
where $C$ is continuous on $\overline{\Lambda_1}\times[0,\pi)$.
\end{theorem}

\begin{proof}
From the representation \eqref{FNphi} and Eqs.~\eqref{cWW} and \eqref{cKs} we obtain the majorization
\begin{multline}
|F_N(\lambda;t,u)|\leq \frac{\exp\big(-u_N \im (t_1+\cdots +t_N)\big)}{(N-1)!}\\ \times\int_{(\R+i\phi_N(t))^{N-1}}ds \frac{\prod_{1\leq m<n\leq N-1}[4\sinh^2(s_m-s_n)]^{\re\lambda}}{\prod_{j=1}^N\prod_{k=1}^{N-1}|2\cosh(t_j-s_k)|^{\re\lambda}}|F_{N-1}(\lambda,s,(u_1-u_N,\ldots,u_{N-1}-u_N))|.
\end{multline}
By the induction assumption, \eqref{FNb} holds true when $N$ is replaced by $N-1$. Applying the resulting bound to the factor $F_{N-1}$ of the integrand, we note that we have $d_{N-1}(s)=0$ and $\phi_{N-1}(s)=\phi_N(t)$. Appealing also to the inequalities \eqref{cIneq} and \eqref{sIneq}, we deduce
\begin{multline}
|F_N(\lambda;t,u)| < C(\lambda,d_N(t))\exp\left(-\left[\phi_N(t)\sum_{j=1}^Nu_j+u_N\sum_{k=2}^{N-1}(\im t_{j_k}-\phi_N(t))\right]\right)\\ \times\int_{\R^{N-1}}ds \prod_{1\leq m<n\leq N-1}(s_m-s_n)\sinh(\re\lambda(s_m-s_n))\prod_{j=1}^N\prod_{k=1}^{N-1}\frac{C^\prime(\lambda,|\im t_j-\phi_N(t)|)}{\cosh(\re \lambda(\re t_j-s_k))},
\end{multline}
with $C^\prime$ continuous on $\overline{\Lambda_1}\times[0,\pi/2)$. Just as in the proof of Prop.~\ref{F3bProp}, we see that the $C^\prime$-product is bounded above by a function $C(\lambda,d_N(t))$ that is continuous on $\overline{\Lambda_1}\times[0,\pi)$. Taking $s\to s/\re \lambda$, we can use \eqref{BN-1}--\eqref{BN-1eval} to compute the remaining integral, yielding the bound \eqref{FNb}.
\end{proof}

The final theorem of this section concerns complex~$u$. As in the previous section, we first present some auxiliary results. To start with, we need a slight change in the notation~\eqref{not3} to prevent ambiguities. Specifically, 
for a given vector $v\in\C^M$, $M\geq 3$, we introduce $V_M\in\C$ and $v^{(M)}\in\C^M$ by
\be\label{notM}
V_M\equiv \frac{1}{M}\sum_{j=1}^Mv_j,\quad v_j^{(M)}\equiv v_j-V_M,\quad j=1,\ldots,M.
\ee
(Thus the $\tilde{t}$ and $\tilde{u}$ of the previous section are now denoted by~$t^{(3)}$ and~$u^{(3)}$.)

To continue, we show inductively that the representation~\eqref{F3s}--\eqref{F3r} generalizes as follows:
\be\label{FNs}
F_N(\lambda;t,u) =  \exp(NiT_NU_N) F_N^r(\lambda;t,u),
\ee
\begin{multline}\label{FNr}
F_N^r(\lambda;t,u)\equiv \frac{1}{(N-1)!}\int_{\R^{N-1}}ds W_{N-1}(\lambda;s)\cK_N^\sharp(\lambda;t^{(N)},s)\\
\times F_{N-1}(\lambda;s,(u_1-u_N,\ldots,u_{N-1}-u_N)),
\end{multline}
where $(\lambda,t,u)\in\overline{\Lambda_1}\times \cA_N^r\times \R^N$, with 
\be\label{cANr}
A_N^r\equiv \{ t\in \C^N \mid \mu_N(t)<\pi/2 \},\ \ 
\mu_N(t)\equiv \max_{j=1,\ldots,N}|\im t^{(N)}_j|.
\ee
To this end we begin by noting that the validity of~\eqref{FNs}--\eqref{FNr} for $N=3$ is implied by   \eqref{F3s}--\eqref{F3r}. Consider  now the integral on the rhs of \eqref{FNphi}, with
\be\label{trest}
\im t_1=\cdots =\im t_N=:\kappa  \Rightarrow d_N(t)=0,\  \phi_N(t)=\im T_N=\kappa.
\ee
Taking $s_j\to s_j+T_N$, it becomes (cf.~\eqref{cWW} and~\eqref{cKs})
\be
\int_{\R^{N-1}}ds W_{N-1}(s)\cK_N^\sharp(t^{(N)},s)F_{N-1}((s_1+T_N,\ldots,s_{N-1}+T_N),(u_1-u_N,\ldots,u_{N-1}-u_N)).
\ee
Making use of the equations~\eqref{FNs}--\eqref{FNr} with $N\to N-1$, which hold true by the induction assumption, it readily follows that
\begin{multline}
F_{N-1}((s_1+T_N,\ldots,s_{N-1}+T_N),(u_1-u_N,\ldots,u_{N-1}-u_N)) =\\ \exp\Big(iT_N\sum_{j=1}^{N-1}(u_j-u_N)\Big)F_{N-1}(s,(u_1-u_N,\ldots,u_{N-1}-u_N)).
\end{multline}
(To check this, observe that $F_{N-1}^r(s,(u_1-u_N,\ldots,u_{N-1}-u_N))$ is invariant under taking $s_j\to s_j+T_N$, $j=1,\ldots, N-1$.) Noting the identity
\be
u_N\sum_{j=1}^Nt_j + T_N\sum_{j=1}^{N-1}(u_j-u_N) = NT_NU_N,
\ee
we now obtain~\eqref{FNs}--\eqref{FNr}, first for $t$ satisfying~\eqref{trest}, and then for~$t\in\cA_N^r$ by analytic continuation.

As the generalization of~\eqref{C2}--\eqref{C2ev} we need the integral
\be\label{CN-1}
C_{N-1}(z,w)\equiv \int_{\R^{N-1}}dv\frac{\prod_{1\le j<k\le N-1}\sinh(v_j-v_k)}{\prod_{j=1}^N\prod_{k=1}^{N-1}\cosh(w_j-v_k)}\exp\big(v_1(z_N-z_1)+\cdots +v_{N-1}(z_N-z_{N-1})\big),
\ee
whose absolute convergence can be ensured by requiring~$|\re(z_N-z_i)|<2$, $i=1,\ldots,N-1$, and $|\im w_j|<\pi/2$, $j=1,\ldots,N$. Assuming also~$z_i\ne z_N$, $i=1,\ldots,N-1$, and $w_j\ne w_k$, $1\le j<k\le N$, its explicit evaluation reads 
\bea\label{CN-1ev}
C_{N-1}(z,w) & = & \prod_{j=1}^{N-1}\frac{\pi}{\sin (\pi (z_{N}-z_j)/2)}\prod_{1\leq j<k\leq N}\frac{1}{\sinh(w_j-w_k)}
\nonumber  \\
&  &  \times\sum_{\tau\in S_{N}}(-)^\tau\exp\Big(\sum_{j=1}^{N-1}w_{\tau(j)}(z_{N}-z_j)\Big),
\eea
as follows from Lemma~C.3 in~\cite{HR13}.

\begin{theorem}
The function~$F_N(\lambda;t,u)$ is holomorphic in
\be\label{cDN}
\cD_N\equiv \{ (\lambda,t,u)\in \Lambda_1\times \cA_N^r\times \C^N \mid  |\im (u_j-u_k)|<2\re \lambda,\ \ j,k=1,\ldots,N \},
\ee
and extends continuously to~$\lambda\in\overline{\Lambda_1}$. Moreover, for all~$(\lambda,t,u)\in\overline{\Lambda_1}\times\cA_N^r\times \C^N $ such that
\be
\im(t_j-t_k)\ne 0,\ \ \   |\im(u_j-u_k)|\in(0,2\re\lambda),\ \ \ 1\le j<k\le N,
\ee
we have
\begin{multline}\label{FNbu}
|F_N(\lambda;t,u)|  <  C(\lambda,\mu_N(t))\exp(- N\im (T_NU_N) )
 \\
 \times \frac{\sum_{\tau\in S_N}(-)^{\tau}\exp\Big(-\sum_{j=1}^N\re \big(t^{(N)}_{\tau(j)}\big)\im \big(u^{(N)}_j\big)\Big) }
{\prod_{1\le j<k\le N}\sin (\pi\im(u_k-u_j)/2\re \lambda)\sinh(\re \lambda\,\re(t_j-t_k))},
\end{multline}
where   $C$ is continuous on $\overline{\Lambda_1}\times[0,\pi/2)$. 
\end{theorem}
\begin{proof}
From~\eqref{FNs} we see that it suffices to estimate~$F_N^r$. By the induction assumption, \eqref{FNbu} is valid for $N\to N-1$. Combining the resulting estimate on~$|F_{N-1}(\lambda;t,u)|$ with the bound~\eqref{cIneq}, and using the generalization of~\eqref{tuid3}, namely,
\be\label{tuidN} 
\sum_{j=1}^{N-1}\re\big(t^{(N)}_{\tau(j)}\big)\im(u_{N}-u_j)=-\sum_{j=1}^N \re\big(t^{(N)}_{\tau(j)}\big)\im\big( u^{(N)}_j\big),
\ee
we obtain  
\begin{multline}\label{FNrb}
|F_N^r(\lambda;t,u)| < \frac{C(\lambda,\mu_N(t))}{\prod_{1\leq j<k\leq N-1}\sin(\pi\im(u_k-u_j)/2\re \lambda)}\\ \times\int_{\R^{N-1}}ds\prod_{1\leq j<k\leq N-1}\frac{[\sinh^2(s_j-s_k)]^{2\re\lambda}}{\sinh(\re\lambda (s_j-s_k))}\frac{1}{\prod_{j=1}^N\prod_{k=1}^{N-1}\cosh(\re\lambda(\re  t^{(N)}_j-s_k))} \\
\times \exp\left( S_{N-1}\sum_{j=1}^{N-1}\im(u_N-u_j)\right) \sum_{\tau\in S_{N-1}}(-)^{\tau}\exp\left(\sum_{j=1}^{N-2}s^{(N-1)}_{\tau(j)}\im(u_{N-1}-u_j)\right),
\end{multline}
where we have used the convention~\eqref{notM}.
 
Since $\re\lambda\ge 1$,  we can invoke the bound~\eqref{sIneq} to deduce that we may replace the sinh-ratio in the integral by~$\sinh(\re\lambda\, (s_j-s_k))$. Then we switch to the variables~\eqref{nvar} (with~$\tilde{t}$ replaced by~$t^{(N)}$, of course).
The result is that the integral becomes equal to $(\re\lambda)^{-N+1}$ times
\begin{multline}\label{FN-1} 
 \int_{\R^{N-1}}dv\frac{\prod_{1\leq j<k\leq N-1}\sinh(v_j-v_k)}{\prod_{j=1}^N\prod_{k=1}^{N-1}\cosh(v_k-w_j)}\exp\left( V_{N-1}\sum_{j=1}^{N-1}(z_N-z_j)\right)\\ \times\sum_{\tau\in S_{N-1}}(-)^{\tau}\exp\left(\sum_{j=1}^{N-2}v^{(N-1)}_{\tau(j)}(z_{N-1}-z_j)\right).
\end{multline} 
Now we fix $\tau\in S_{N-1}$ and consider the corresponding summand on the rhs of \eqref{FN-1}. Changing variables $v_j\to v_{\tau^{-1}(j)}$ and using antisymmetry of the $\sinh$-product, we obtain the summand for which $\tau$ is the identity permutation. Hence we may replace the sum by $(N-1)!$ times the $\tau={\rm id}$ summand. We simplify the resulting product of exponential functions by using the readily checked identity
\be
V_{N-1}\sum_{j=1}^{N-1}(z_N-z_j) + \sum_{j=1}^{N-2}v^{(N-1)}_j(z_{N-1}-z_j) = \sum_{j=1}^{N-1}v_j(z_N-z_j).
\ee
Comparing to~\eqref{CN-1}, we deduce that the integral equals~$(N-1)!\,C_{N-1}(z,w)$. Next we use the evaluation~\eqref{CN-1ev} and reverse the substitution~\eqref{nvar}. Using  once more~\eqref{tuidN}, we now obtain the estimate~\eqref{FNbu}. From this estimate the asserted holomorphy and continuity properties readily follow. 
\end{proof}

\section{The Heckman--Opdam hypergeometric function}
In this final section we establish the precise connection between the function~$F_N$ and the Heckman--Opdam hypergeometric function associated with the root system~$A_{N-1}$. As a first and crucial step, we decompose $F_N(t,u)$ into a product of a center-of-mass factor and a function $F_N^r(t,u)$ depending only on the differences $t_j-t_{j+1}$ and $u_j-u_{j+1}$ with $j=1,\ldots,N-1$, cf.~\eqref{FNs}--\eqref{FNr}. The point is that when we express the eigenfunction property \eqref{eNEq} in terms of the function $F_N^r$, then we are able to compare the result with the system of hypergeometric differential equations introduced in \cite{HO87}. The details now follow.

We begin by obtaining the pertinent eigenfunction property for $F_N^r$. First, we act with $:\hat{\Sigma}_k^{(N)}(\cL+\cE)(t):$ on \eqref{FNs} and shift the PDO past the exponential factor using
\be
:\hat{\Sigma}_k^{(N)}(\cL+\cE)(t):\exp\left(NiT_NU_N\right) = \exp\left(NiT_NU_N\right):\hat{\Sigma}_k^{(N)}(\cL+\cE+U_N\mathbf{1}_N)(t):.
\ee
Second, we expand the rhs as in \eqref{Sexp} with $u_n\to U_N$ and insert the resulting expression in the lhs of \eqref{eNEq}. Third, we expand $S_k^{(N)}(u)$ in the rhs of \eqref{eNEq} according to
\be
\begin{split}
S_k^{(N)}(u) &= S_k^{(N)}(u_1^{(N)}+U_N,\ldots,u_N^{(N)}+U_N)\\ &= \sum_{l=0}^kU_N^l\binom{N-k+l}{l}S_{k-l}^{(N)}\big(u^{(N)}\big).
\end{split}
\ee
Fourth, comparing powers of $U_N$, we obtain
\be\label{FNrEq}
:\hat{\Sigma}_k^{(N)}(\cL+\cE)(t):F_N^r(t,u) = S_k^{(N)}\big(u^{(N)}\big)F_N^r(t,u),\quad k=1,\ldots,N.
\ee

The commutative PDO algebra generated by $:\hat{\Sigma}_k^{(N)}(\cL+\cE)(t):$,  $k=1,\ldots,N$, contains in particular the second order PDO
\be\label{LN}
\begin{split}
L_2(t) &\equiv \frac{1}{4}\sum_{j=1}^N\partial_{t_j}^2 + \frac{\lambda}{2}\sum_{1\leq j<k\leq N}\coth(t_j-t_k)(\partial_{t_j}-\partial_{t_k})\\ &= \frac{1}{2}:\hat{\Sigma}_2^{(N)}(\cL+\cE)(t): - \frac{1}{4}\left(:\hat{\Sigma}_1^{(N)}(\cL+\cE)(t):\right)^2 - \left(\rho,\rho\right),
\end{split}
\ee
where 
\be
\rho \equiv \frac{\lambda}{2}(N-1,N-3,\ldots,-N+3,-N+1),
\ee
and $(\cdot,\cdot)$ denotes the standard bilinear form on $\C^N$. Also, combining~\eqref{FNrEq} and \eqref{LN} we obtain the eigenvalue equation
\be\label{LNEq}
L_2(t)F_N^r(t,u) = \left(iu^{(N)}/2-\rho,iu^{(N)}/2+\rho\right)F_N^r(t,u).
\ee

The above equations \eqref{LN} and \eqref{LNEq} should now be compared with Eqs.~(2.6) and (3.12) in~\cite{HO87} for the root system choice~$A_{N-1}$. The latter equations involve a coupling parameter~$k$ and eigenvalue vector~$\tilde{\lambda}$, as well as a quantity~$h$ ranging over the~$SL(N,\C)$-torus~$H$.  (We have added the tilde to the notation $\lambda$ employed in~\cite{HO87}, so as to prevent confusion with our coupling parameter~$\lambda$.) 
To make an explicit comparison possible, we fix the simple roots in $A_{N-1}$ to be $e_j-e_{j+1}$, $j=1,\ldots,N-1$ (with~$e_1,\ldots,e_N$ the standard basis in~$\C^N$), and view~$h$ as a diagonal~$N\times N$ matrix with $\det(h)=1$. Then
  we  arrive at the identifications
\be\label{ident}
k=\lambda,\quad (\tilde{\lambda}_1,\ldots,\tilde{\lambda}_N) = (iu_1/2,\ldots,iu_N/2),\quad h(t) = \diag\left(e^{2t_1},\ldots,e^{2t_N}\right),
\ee
with $U_N=T_N=0$.  

Imposing these identifications, the  eigenvalue equations~\eqref{FNrEq} (with~$(t,u)\in\cA_N\times\R^N$   satisfying $T_N=U_N=0$) give rise to the system of hypergeometric differential equations in Def.~2.13 of~\cite{HO87}. Indeed, the PDO $:\hat{\Sigma}_1^{(N)}(\cL+\cE)(t):$ acts as the zero operator on~$F_N^r(t,u)$, whereas the remaining~$N-1$ PDOs $:\hat{\Sigma}_k^{(N)}(\cL+\cE)(t):$  yield a generating set for the algebra~$\D$ in~\cite{HO87} corresponding to~$A_{N-1}$.

Heckman and Opdam constructed a basis of~$N!$ solutions ``at infinity'' of the system of hypergeometric differential equations and singled out a special linear combination~$F(\tilde{\lambda},k;h(t))$. They showed that it extends to a holomorphic function in the~$N-1$ difference variables $t_j-t_{j+1}$ on a neighbourhood of the origin and that (generically) the function $F(\tilde{\lambda},k;h(t))$ is up to a constant characterized by this property. Moreover, they conjectured that with a specific normalization their function satisfies $F(\tilde{\lambda},k;h(0))=1$, cf.~Conjecture~6.11 in~\cite{HO87}. This conjecture was later proved by Opdam \cite{Opd93}. From this state of affairs we readily deduce  the following relationship.

\begin{proposition}
Let $\lambda\in\overline{\Lambda_1}$ and let $u\in\R^N$ satisfy $U_N=0$. For all $t\in \cA_N$ with $T_N=0$, we have
\be
F(iu/2,\lambda;h(t)) = \frac{F_N^r(\lambda;t,u)}{F_N^r(\lambda;0,u)}.
\ee
\end{proposition}

\bibliographystyle{amsalpha}

\end{document}